\documentclass[preprint,nocopyrightspace]{sigplanconf}

\newif\iflong\longtrue

\newcommand\appref[1]{Appendix~\ref{#1}}

\makeatletter
\newcommand\published{%
  \thanks{Extended version of the paper with the same title, to
    appear in Proc. \textsl{\@conferencename}, \@conferenceinfo, and to be 
    available at\\ \url{http://dx.doi.org/10.1145/\@doi}}}
\makeatother

\usepackage{graphics, graphicx}
\usepackage{color,amsmath, amsfonts,amsmath,amssymb,amsthm,setspace}
\usepackage[curve,matrix,arrow,cmtip]{xy} 
\usepackage{tikz}
\usepackage{mathpartir}
\usepackage{stmaryrd}
\usepackage{bbold}
\usepackage{pdfnotiz} 
\usepackage{wrapfig}
\usepackage{enumerate}

\definecolor{citecolor}{rgb}{0.0,0.4,0.0}
\definecolor{urlcolor}{rgb}{0.0,0.0,0.4}
\definecolor{linkcolor}{rgb}{0.0,0.0,0.4}
\usepackage[breaklinks=true,colorlinks,linkcolor=linkcolor,urlcolor=urlcolor,citecolor=citecolor]{hyperref} 

\usetikzlibrary{matrix}
\usetikzlibrary{positioning,arrows,decorations.markings}
\tikzstyle{vecArrow} = [thick, decoration={markings,mark=at position
   1 with {\arrow[semithick]{open triangle 60}}},
   double distance=1.4pt, shorten >= 5.5pt,
   preaction = {decorate},
   postaction = {draw,line width=1.4pt, white,shorten >= 4.5pt}]
\tikzstyle{innerWhite} = [semithick, white,line width=1.4pt, shorten >= 4.5pt]

\makeatletter
\newtheorem*{rep@theorem}{\rep@title}
\newcommand{\newreptheorem}[2]{%
\newenvironment{rep#1}[1]{%
 \def\rep@title{#2 \ref{##1}}%
 \begin{rep@theorem}}%
 {\end{rep@theorem}}}
\makeatother

\newtheorem{theorem}{Theorem}
\newreptheorem{theorem}{Theorem}
\newtheorem{lemma}{Lemma}
\newreptheorem{lemma}{Lemma}
\newtheorem{proposition}{Proposition}
\newreptheorem{proposition}{Proposition}

\newtheorem{corollary}{Corollary}
\newtheorem{definition}{Definition}
\newtheorem{example}{Example}

\newcommand{\powe}{\mathcal{P}}
\newcommand{\pow}{\mathcal{P}_{\omega}}
\newcommand{\lift}[1]{\overline{#1}}
\newcommand{\laxlift}[1]{\Rel(#1)^{\sqsubseteq}}
\newcommand{\wlift}[1]{\Rel(#1)^{[\supseteq \subseteq]}}

\newcommand{\A}{\mathbb A}
\newcommand{\B}{\mathcal B}
\newcommand{\C}{\mathcal C}

\renewcommand{\P}{\ensuremath{\mathcal E}}

\newcommand{\R}{\mathbb R}

\newcommand{\Cat}{\mathsf{Cat}} 
\newcommand{\Set}{\mathsf{Set}} 
 \newcommand{\Rel}{\mathsf{Rel}}
\newcommand{\Pred}{\mathsf{Pred}} 
\newcommand{\Alg}{\mathsf{Alg}} 
\newcommand{\Pre}{\mathsf{Pre}}
\newcommand{\Fib}{\mathsf{Fib}}

\newcommand{\T}{\mathbb T}
\newcommand{\Id}{\mathrm{Id}}

\newcommand{\pred}[2]{#1_{#2}} %coinductive predica for \xi^* \after \lift{F}
\newcommand{\cpred}[2]{\pred{\Rel(#1)}{#2}}

\newcommand{\Nom}{\mathsf{Nom}}
\newcommand{\Pf}{\powe_{\omega}}
\newcommand{\At}{\mathbb{A}}
\newcommand{\Sub}{\mathsf{Sub}}

\newcommand{\Cgr}{\mathit{Cgr}}
\newcommand{\Tra}{\mathit{Trn}}
\newcommand{\Sym}{\mathit{Sym}}
\newcommand{\Ref}{\mathit{Rfl}}
\newcommand{\Beh}{\mathit{Bhv}}
\newcommand{\Con}{\mathit{Ctx}}
\newcommand{\Slf}{\mathit{Slf}}
\newcommand{\Eqv}{\mathit{Eqv}}

\newcommand{\Ra}{\Rightarrow}

 \renewcommand{\bar}{\overline}
 \newcommand{\op}{\mathit{op}}
 \newcommand\restr[2]{{% we make the whole thing an ordinary symbol
  \left.\kern-\nulldelimiterspace % automatically resize the bar with \right
  #1 % the function
  \vphantom{\big|} % pretend it's a little taller at normal size
  \right|_{#2} % this is the delimiter
  }}
\newcommand{\pullbackcorner}[1][dr]{\save*!/#1-1.5pc/#1:(-1.2,1.2)@^{|-}\restore}
  
\newcommand\dfa[1]{{%
  \newcommand\state[1]{##1}%
  \newcommand\fstate\overline%
  #1}}
\newcommand\nfa\dfa

\def\tr#1{\stackrel{#1}{\to}}

\makeatletter
\newcommand{\superimpose}[2]{%
  {\ooalign{$#1\@firstoftwo#2$\cr\hfil$#1\@secondoftwo#2$\hfil\cr}}}
\makeatother

\newcommand{\diamondtau}{\langle\tau\rangle}%{\mathpalette\superimpose{{\Diamond}{\tau}}}

\newcommand{\ldag}{\lambda^\dagger}
\newcommand{\ldagbar}{\bar{\lambda}^\dagger}

\newcommand{\id}{\mathrm{id}}

\begin{document}

\setlength{\pdfpageheight}{\paperheight}
\setlength{\pdfpagewidth}{\paperwidth}

\conferenceinfo{CSL-LICS 2014}{July 14--18, 2014, Vienna, Austria}
\copyrightyear{2014}
\copyrightdata{978-1-4503-2886-9}
\doi{2603088.2603149}

\title{Coinduction Up-To in a Fibrational Setting\published}

\authorinfo{Filippo Bonchi \and Daniela Petri{\c s}an
\and Damien Pous\thanks{The first three authors acknowledge support from the ANR projects 2010-BLAN-0305
PiCoq and 12IS02001 PACE.}}
{LIP, CNRS, INRIA, ENS Lyon, \\Universit\'e de Lyon, UMR 5668}
{\{filippo.bonchi,daniela.petrisan,damien.pous\}\\@ens-lyon.fr}
\authorinfo{Jurriaan Rot\thanks{The research of this author has been funded by the Netherlands
Organisation for Scientific Research (NWO), CoRE project,
dossier number: 612.063.920.}}
{LIACS - Leiden University, CWI}
{j.c.rot@liacs.leidenuniv.nl}

\maketitle

\begin{abstract}
  Bisimulation up-to enhances the coinductive proof method for
  bisimilarity, providing efficient proof techniques for checking
  properties of different kinds of systems.
  We prove the soundness of such techniques in a fibrational setting,
  building on the seminal work of Hermida and Jacobs.  This allows us
  to systematically obtain up-to techniques not only for bisimilarity
  but for a large class of coinductive predicates modelled as
  coalgebras.
  By tuning the parameters of our framework, we obtain novel
  techniques for unary predicates and nominal automata, a variant of
  the GSOS rule format for similarity, and a new categorical treatment
  of weak bisimilarity.
\end{abstract}




\section{Introduction}

\subsection{Coinduction up-to}
\label{ssec:intro:coinduction-upto}

The rationale behind coinductive up-to techniques is the following.
Suppose you have a characterisation of an object of interest as a
greatest fixed-point. For instance, behavioural equivalence in CCS is
the greatest fixed-point of a monotone function $B$ on relations,
describing the standard bisimulation game. This means that to prove
two processes equivalent, it suffices to exhibit a relation $R$ that
relates them, and which is a \emph{$B$-invariant}, i.e., $R\subseteq
B(R)$. Such a task can however be painful or inefficient, and one
could prefer to exhibit a relation which is only a $B$-invariant
\emph{up to some function} $A$, i.e., $R\subseteq B(A(R))$.

Not every function $A$ can safely be used: $A$ should be \emph{sound}
for $B$, meaning that any $B$-invariant up to $A$ should be contained
in a $B$-invariant. Instances of sound functions for behavioural
equivalence in process calculi usually include transitive closure,
context closure and congruence closure.
The use of such techniques dates back to Milner's work on
CCS~\cite{Milner89}; a famous example of an unsound technique is
that of weak bisimulation up to weak bisimilarity. Since then,
coinduction up-to proved useful, if not essential, in numerous proofs
about concurrent systems (see~\cite{PS12} for a list of
references); it has been used to obtain decidability
results~\cite{Caucal90}, and more recently to improve standard
automata algorithms~\cite{bp:popl13:hkc}.

The theory underlying these techniques was first developed by
Sangiorgi~\cite{San98MFCS}. It was then reworked and generalised by
one of the authors to the abstract setting of complete
lattices~\cite{pous:aplas07:clut,PS12}. The key observation there
is that the notion of soundness is not compositional: the
composition of two sound functions is not necessarily sound itself. 
The main solution to this problem consists in restricting to
\emph{compatible} functions, a subset of the sound functions which
enjoys nice compositionality properties and contains most
of the useful techniques.

An illustrative example of the benefits of a modular theory is the
following: given a signature $\Sigma$, consider the \emph{congruence
  closure} function, that is, the function $\Cgr$ mapping a relation
$R$ to the smallest congruence containing $R$.  This function has
proved to be useful as an up-to technique for language equivalence of
non-deterministic automata~\cite{bp:popl13:hkc}.  It can be decomposed
into small pieces as follows: $\Cgr=\Tra\circ\Sym\circ\Con\circ\Ref$,
where $\Tra$ is the transitive closure, $\Sym$ is the symmetric
closure, $\Ref$ is the reflexive closure, and $\Con$ is the context
closure associated to $\Sigma$.  Since compatibility is preserved by
composition (among other operations), the compatibility of $\Cgr$
follows from that of its smaller components.  In turn, transitive
closure can be decomposed in terms of relational composition, and
context closure can be decomposed in terms of the smaller functions
that close a relation with respect to $\Sigma$ one symbol at a
time. Compatibility of such functions can thus be obtained in a
modular way.

A key observation in the present work is that when we move to a
coalgebraic presentation of the theory, compatible functions
generalise to functors equipped with a distributive law
(Section~\ref{sec:compatible-functors}). 

\subsection{Fibrations and coinductive predicates}
\label{ssec:intro:coinductive-predicates}

Coalgebras are a tool of choice for describing state based systems:
given a functor $F$ determining its type (e.g., labelled transition
systems, automata, streams), a system is just an $F$-coalgebra
$(X,\xi)$.  When $F$ has a final coalgebra $(\Omega,\omega)$, this
gives a canonical notion of behavioural
equivalence~\cite{Jacobs:coalg}:
\begin{align*}
  \xymatrix{
    X\ar[d]_\xi\ar[r]^{\llbracket\cdot\rrbracket}&\Omega\ar[d]^\omega\\
    FX\ar[r]^{F\llbracket\cdot\rrbracket}&F\Omega}
\end{align*}
two states $x,y\in X$ are equivalent % ($x\sim y$)
if they are mapped to the same element in the final coalgebra.

When the functor $F$ preserves weak pullbacks---which we shall assume
throughout this introductory section for the sake of
simplicity---behavioural equivalence can be characterised
coinductively using Hermida-Jacobs
bisimulations~\cite{Hermida97structuralinduction,Staton11}: given an $F$-coalgebra $(X,\xi)$, behavioural
equivalence is the largest $B$-invariant for a monotone function $B$ on $\Rel_X$, the poset of binary relations over $X$.
This function $B$ can be decomposed as
\begin{equation*}
%  \tag{$\dagger$}
  B~\triangleq~\xi^*\circ \Rel(F)_X\colon \Rel_X\to \Rel_X
\end{equation*}
Let us explain the notations used here. We consider the category $\Rel$ whose objects are relations $R{\subseteq}X^2$ 
and morphisms from  $R{\subseteq}X^2$ to $S{\subseteq}Y^2$ are maps from $X$ to $Y$
sending pairs in $R$ to pairs in $S$. For each set $X$ the poset $\Rel_X$ of binary relations over $X$ is a subcategory of $\Rel$, also called the fibre over $X$.  The functor $F$ has a canonical lifting to $\Rel$, denoted by $\Rel(F)$. This lifting restricts to a functor  $\Rel(F)_X \colon \Rel_X \to \Rel_{FX}$, which in this case is just a  monotone function between posets. The monotone function $\xi^* \colon \Rel_{FX} \to \Rel_X$ is the \emph{inverse image} of the coalgebra $\xi$ mapping a relation $R \subseteq (FX)^2$
to $(\xi \times \xi)^{-1}(R)$.

To express other predicates than behavioural equivalence,  one can take arbitrary \emph{liftings} of $F$ to $\Rel$, different from the canonical one. Any lifting $\overline{F}$ yields a functor $B$ defined as
\begin{equation*}
  \tag{$\dagger$}
  B~\triangleq~\xi^*\circ \overline F_X\colon \Rel_X\to \Rel_X
\end{equation*}
The final coalgebra, or greatest fixed-point for such a $B$ is called a \emph{coinductive
predicate}~\cite{Hermida97structuralinduction, HasuoCKJ:coindFib}.
%Indeed, $B$-coalgebras in $\Rel_X$ are invariants, and the final $B$-coalgebra is the greatest invariant.  
By taking appropriate $\overline F$, one can obtain, for instance, various behavioural preorders: similarity on
labelled transition systems (LTSs), language inclusion on automata, or
lexicographic ordering of streams.

This situation can be further generalised using \emph{fibrations}. We refer the reader to the first chapter of~\cite{Jacobs:fib} for a gentle introduction, or to Section~\ref{sec:prelim} for succinct definitions. The functor $p \colon \Rel \to \Set$ mapping a relation $R\subseteq X^2$ to its support set $X$ is a fibration, where the inverse image $\xi^*$ is just the \emph{reindexing functor} of $\xi$.
By choosing a different fibration than $\Rel$, one can obtain coinductive characterisations of
objects that are not necessarily binary relations, e.g., unary
predicates like divergence, ternary relations, or metrics.

Our categorical generalisation of compatible functions provides a
natural extension of this fibrational framework with a systematic
treatment of up-to techniques: we provide functors (i.e., monotone
functions in the special case of the $\Rel$ fibration) that are
compatible with those functors $B$ corresponding to coinductive
predicates.

For instance, when the chosen lifting $\overline F$ is a
\emph{fibration map}, the functor corresponding to a technique called
``up to behavioural equivalence'' is compatible
(Theorem~\ref{theo:beh}). The canonical lifting of a functor is always
such a fibration map, so that when $F$ is the functor for LTSs, we
recover the soundness of the very first up-to technique from the literature, namely
``bisimulation up to bisimilarity''~\cite{Milner89}.
One can also check that another lifting of this same functor but in
another fibration yields the divergence predicate, and is a
fibration map.  We thus obtain the validity of the ``divergence up
to bisimilarity'' technique.

\subsection{Bialgebras and up to context}
\label{ssec:intro:bialgebras}

Another important class of techniques comes into play when considering
systems with an algebraic structure on the state space (e.g., the syntax of a process
calculus). A minimal requirement for such systems usually is that
behavioural equivalence should be a congruence. In the special case of
bisimilarity on LTSs, several rule formats have been proposed to
ensure such a congruence property~\cite{AFV}. At the categorical
level, the main concept to study such systems is that of
\emph{bialgebras}. Assume two endofunctors $T,F$ related by a
distributive law $\lambda\colon TF\Rightarrow FT$. A
$\lambda$-bialgebra consists in a triple $(X,\alpha,\xi)$ where
$(X,\alpha)$ is a $T$-algebra, $(X,\xi)$ is an $F$-coalgebra, and a
diagram involving $\lambda$ commutes. It is well known that in such a
bialgebra, behavioural equivalence is a congruence with respect to
$T$~\cite{DBLP:conf/lics/TuriP97}. This is actually a generalisation
of the fact that bisimilarity is a congruence for all GSOS
specifications~\cite{BloomCT88:GSOS}: GSOS specifications are in
one-to-one correspondence with distributive laws between the
appropriate functors~\cite{DBLP:conf/lics/TuriP97, Bartels03}.

This congruence result can be strengthened into a compatibility
result~\cite{RotBBPRS}: in any $\lambda$-bialgebra, the
contextual closure function that corresponds to $T$ is compatible for
behavioural equivalence. By moving to fibrations, we generalise this
result so that we can obtain up to context techniques for arbitrary
coinductive predicates: unary predicates like divergence, by using
another fibration than $\Rel$; but also other relations than
behavioural equivalence, like the behavioural preorders mentioned
above, or weak bisimilarity.

The technical device we need to establish this result is that of
\emph{bifibrations}, fibrations $p$ whose opposite functor $p^\op$ is
also a fibration. We keep the running example of the $\Rel$ fibration
for the sake of clarity; the results are presented in full generality
in the remaining parts of the paper.
% $p:\P\to \B$ such that $p^\op:\P^\op\to\B^\op$ is also a fibration.
In such a setting, any morphism $f\colon X\to Y$ in $\Set$ has a
\emph{direct image} $\coprod_f \colon \Rel_X\to\Rel_Y$. 
% which actually is a left adjoint to the cartesian lifting
% $f^\star:\Rel_Y\to\Rel_X$. 
Now given an algebra $\alpha \colon TX\to X$ for a functor $T$ on $\Set$,
any lifting $\lift{T}$ of $T$ gives rise to a functor on the fibre
above $X$, defined dually to~$(\dagger)$:
\begin{equation*}
  \tag{$\ddagger$}
  C~\triangleq~\textstyle{\coprod}_\alpha\circ\overline T_X\colon \Rel_X\to \Rel_X
\end{equation*}
When we take for $\overline T$ the canonical lifting of $T$ in $\Rel$,
then $C$ is the contextual closure function corresponding to the
functor $T$. We shall see that we sometimes need to consider
variations of the canonical lifting to obtain a compatible up-to
technique (e.g., up to ``monotone'' contexts for checking language
inclusion of weighted automata---Section~\ref{ssec:weighted}).

Now, starting from a $\lambda$-bialgebra $(X,\alpha,\xi)$,
and given two liftings $\overline T$ and $\overline F$ of $T$ and $F$,
respectively, the question is whether the above functor $C$ is
compatible with the functor $B$ defined earlier in~$(\dagger)$. The
simple condition we give in this paper is the following: the
distributive law $\lambda\colon TF\Rightarrow FT$ should lift to a
distributive law $\overline\lambda\colon \overline T\,\overline F\Rightarrow
\overline F\,\overline T$ (Theorem~\ref{thm:big-2-cell}).

This condition is always satisfied in the bifibration $\Rel$, when
$\overline T$ and $\overline F$ are the canonical liftings of $T$ and
$F$. Thus we obtain as a corollary the compatibility of bisimulation
of up to context in $\lambda$-bialgebras, which is the main result
from~\cite{RotBBPRS}---soundness was previously observed by Lenisa et
al.~\cite{Lenisa99,LenisaPW00} and then Bartels~\cite{Bartels03}.

The present work allows us to go further in several directions, as
illustrated below.

\subsection{Contributions and Applications}

The main contribution of this paper is the abstract framework
developed in Section~\ref{sec:upto:fibrations}; it allows us to
derive the soundness of a wide range of both novel and well-established
up-to techniques for arbitrary coinductive
predicates. Sections~\ref{sec:examples} and~\ref{sec:compositional}
are devoted to several such applications, which we describe now.

%To illustrate its expressive power we provide various applications. 

When working in the predicate fibration on $\Set$, one can
characterise some formulas from modal logic as coinductive predicates
(see~\cite{DBLP:journals/cj/CirsteaKPSV11} for an account of
coalgebraic modal logic).  Our framework allows us to introduce up-to
techniques in this setting: we consider the formula $\nu x.\diamondtau
x$ in Section~\ref{ssec:divergence}, and we provide a technique called
``divergence up to left contexts and behavioural equivalence''. We use
it to prove divergence of a simple process using a finite invariant,
while the standard method requires an infinite one.

One can also change the base category: by considering the fibration of
equivariant relations over nominal sets, we show how to obtain up-to
techniques for language equivalence of non-deterministic nominal
automata~\cite{BojanczykKL11}. In Section~\ref{ssec:nominal-automata},
these techniques allow us to prove the equivalence of two nominal
automata using an orbit-finite relation, where the standard method
would require an infinite one (recall that the determinisation of a
nominal automaton is not necessarily orbit-finite).

Another benefit of the presented theory is modularity w.r.t.\ the
liftings chosen to define coinductive predicates: two liftings can be
composed, and we give sufficient conditions for deriving compatible
functors for the composite lifting out of compatible functors for its
sub-components (Section~\ref{sec:compositional}).  We give two
examples of such a situation: similarity, and weak bisimilarity on
LTSs.

By using Hughes and Jacobs' definition of similarity~\cite{HughesJ04},
we obtain that for ``up to context'' to be compatible it suffices to
start from a \emph{monotone} distributive law
(Section~\ref{ssec:compositional:simulation}). In the special case of
LTSs, this monotonicity condition amounts to the \emph{positive GSOS}
rule format~\cite{MS10}: GSOS~\cite{BloomCT88:GSOS} without negative
premises.

In Section~\ref{ssec:weak} we propose a novel characterisation of weak
bisimilarity on LTSs, that fits into our framework. This allows us to
give a generic condition for ``up to context'' to be compatible (and
hence weak bisimilarity to be a congruence). In particular, this
condition rules out the sum operation from CCS, which is well known
not to preserve weak bisimilarity.

\section{Preliminaries}
\label{sec:prelim}
We refer the reader to~\cite{Jacobs:fib} for background on fibrations
and recall here basic definitions. 

\begin{definition}
A functor $p\colon\P\to\B$ is called a \emph{fibration} when for every
morphism $f\colon X\to Y$ in $\B$ and every $R$ in $\P$ with $p(R)=Y$
there exists a map $\widetilde{f_R}\colon f^*(R)\to R$ such that
$p(\widetilde{f_R})=f$ satisfying the universal property: For all maps $g:Z\to X$ in $\B$ and $u\colon Q\to R$ in $\P$ sitting above $fg$ (i.e., $p(u)=fg$) there is a unique map $v\colon Q\to f^*(R)$ such that $u=\widetilde{f_R}v$ and $p(v)=g$.
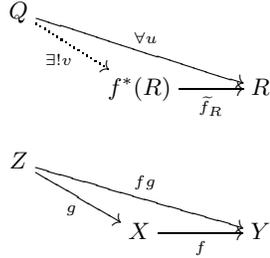
\[
\xymatrix @R=1.6em
{
          Q\ar@{..>}[dr]_-{\exists ! v}\ar[rrd]^-{\forall u} & &  \\
          & f^*(R)\ar[r]_-{\widetilde{f}_{R}} & R \\
          Z\ar[rd]_{g}\ar[rrd]^{fg} & &\\
          &X\ar[r]_-{f} & Y\\
}
\]
\end{definition}

For $X$ in $\B$ we denote by $\P_X$ the \emph{fibre} above $X$, i.e., the
subcategory of $\P$ with objects mapped by $p$ to $X$ and arrows sitting above the identity on $X$.

A map $\widetilde{f}$ as above is called a \emph{Cartesian lifting} of $f$ and is unique up to isomorphism.  If we make a choice of Cartesian liftings, the association $R\mapsto f^*(R)$ gives rise to the 
so-called \emph{reindexing functor} $f^*\colon \P_Y\to\P_X$.

The fibrations considered in this paper are \emph{bicartesian} (both $\P$ and $\B$ have a bicartesian structure strictly preserved by $p$) and \emph{split}, i.e., the reindexing functors behave well with respect to composition and identities: $(1_X)^*=1_{\P_X}$ and $(f\circ g)^*=g^*\circ f^*$.

%\begin{definition}
A functor $p\colon \P\to\B$ is called a \emph{bifibration} if both
$p\colon \P\to\B$ and $p^\op\colon \P^\op\to\B^\op$ are fibrations.
%\end{definition}
A fibration $p\colon \P\to\B$ is a bifibration if and only if each
reindexing functor $f^*\colon \P_Y\to\P_X$ has a left adjoint
$\coprod_f\dashv f^*$,
see~\cite[Lemma~9.1.2]{Jacobs:fib}.

\begin{example}
\label{ex:pred-fib}
Let $\Pred$ be the category of predicates: objects are pairs of sets $(P,X)$ with $P\subseteq X$ and morphisms $f\colon(P,X)\to(Q,Y)$ are arrows $f\colon X\to Y$ that can be restricted to $\restr{f}{P}\colon P\to Q$. 

Similarly, we can consider the category $\Rel$ whose objects are pairs of sets $(R,X)$ with $R\subseteq X^2$ and morphisms $f\colon (R,X)\to(S,Y)$ are arrows $f\colon X\to Y$ such that $f\times f$ can be restricted to $\restr{f\times f}{R}\colon R\to S$. 

The functors mapping predicates, respectively, relations to their
underlying sets are bifibrations. The fibres $\Pred_X$ and $\Rel_X$
sitting above $X$ are the posets of subsets of $X$, respectively
relations on $X$, ordered by inclusion. The reindexing functors are
given by inverse image and their left adjoints by direct image.
\end{example}

%\begin{definition}
Given fibrations $p\colon \P\to\B$ and $p'\colon \P'\to\B$ and
$F\colon \B\to\B$, we call $\lift{F}\colon \P\to\P'$ a \emph{lifting}
of $F$ when $p'\lift{F}=Fp$. Notice that a lifting $\lift{F}$
restricts to a functor between the fibres $\lift{F}_X\colon
\P_X\to\P_{FX}$. When the subscript $X$ is clear from the context we
will omit it.

A \emph{fibration map} between $p\colon \P\to\B$ and $p'\colon \P'\to\B$ is a pair $(\lift{F},F)$ such that  $\lift{F}$ is a lifting of $F$ that preserves the Cartesian liftings:
$(Ff)^*\lift{F} =\lift{F}f^*$ for any $\B$-morphism $f$. We denote by $\Fib(\B)$ the category of fibrations with base $\B$. 
%\end{definition}

\begin{example}
  A $\Set$-endofunctor $T$ has a \emph{canonical} relation lifting $\Rel(T)\colon \Rel\to\Rel$. Represent  $R\in\Rel_X$ as a jointly mono span  $X\leftarrow R\to X$ and apply $T$. Then  $\Rel(T)(R)$ is obtained by factorising the induced map $TR\to TX\times TX$. When $T$ preserves weak pullbacks, $(\Rel(T),T)$ is a fibration map (see e.g.~\cite{HughesJ04}).
\end{example}

\section{Compatible Functors}\label{sec:compatible-functors}
Given two monotone functions $A,B\colon \C \to \C$ on a complete
lattice $\C$, $A$ is said to be $B$-compatible if $AB\subseteq BA$.
In \cite[Theorem 6.3.9]{PS12}, it is shown that any $B$-compatible
function $A$ is sound, that is, it can be used as an up-to technique:
every $B$-invariant up to $A$ is included in a $B$-invariant.

This result is an instance of a more general fact which holds in any
category $\C$ with countable coproducts and for any pair of
endofunctors $A,B$ equipped with a distributive law $\gamma \colon AB
\Ra BA$. Indeed, following the proof of
\cite[Theorem~3.8]{Bartels03}, for any
$BA$-coalgebra $\xi$ (that is a $B$-invariant up to $A$) one can find
a $B$-coalgebra $\zeta$ (that is a $B$-invariant) making the next
diagram commutative.
$$
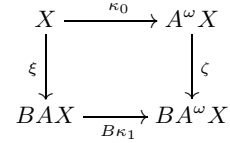
\xymatrix {
  X \ar[d]_{\xi}  \ar[r]^{\kappa_0} & A^\omega X  \ar[d]^{\zeta} \\
  BAX \ar[r]_{B\kappa_1} & BA^\omega X }$$ (Here $A^\omega$ denotes
the coproduct $\coprod_{i\le\omega}A^i$ of all finite iterations of
$A$ and $\kappa_0$, $\kappa_1$ are the injections of $X$
and $AX$ respectively, into $A^\omega X$. Alternatively, we can replace the countable coproduct $A^\omega$ by the free monad on $A$, assuming the latter exists. 
In this case, the result is an instance of the generalized powerset construction~\cite{SilvaBBR10}.)
 
Similarly, that compatible functions
preserve bisimilarity~\cite[Lemma 6.4.3]{PS12} is an instance of the well-known fact~\cite{DBLP:conf/lics/TuriP97} that a final $B$-coalgebra $\nu B$ lifts to a final $\gamma$-bialgebra for $\gamma\colon AB\Ra BA$. When $\C$ is a lattice, this entails that $A(\nu
B) \subseteq \nu B$.
% , thus, the coinductive predicate $\nu B$ is closed w.r.t.\ $A$.
%
For instance, if $B$ is a predicate for bisimilarity and $A$ is the
congruence closure function, we obtain that bisimilarity is a
congruence whenever the congruence closure function is compatible.

As discussed in the Introduction, the main interest in compatible
functions comes from their nice compositionality properties. This leads
us to define compatibility of arbitrary functors of type $\C\to\C'$
rather than just endofunctors.

\begin{definition}\label{def:compatible-functor}
  Consider two endofunctors $B\colon \C\to\C$ and $B':\C'\to\C'$. We say
  that a functor $A\colon \C\to\C'$ is \emph{$(B,B')$-compatible} when
  there exists a natural transformation $\gamma\colon AB\Ra B'A$.
\end{definition}

Notice that the pair $(A,\gamma)$ is a morphism between endofunctors
$B$ and $B'$ in the sense of~\cite{LenisaPW00}. Since the examples
dealt with in this paper involve only poset fibrations, we will omit
the natural transformation $\gamma$ from the notation.  Moreover,
given an endofunctor $B:\C\to\C$, we will simply write that
$A:\C^n\to\C^m$ is $B$-compatible, when $A$ is $(B^n,B^m)$-compatible.

This definition makes it possible to use the internal notions of
product and pairing to emphasise the compositionality aspect.  For
instance, coproduct becomes a compatible functor by itself, rather
than a way to compose compatible functors.

\begin{proposition}\label{prop:modularity}
  Compatible functors are closed under the following constructions:
  \begin{enumerate}[(i)]
  \item\label{it:compo} composition: if $A$ is $(B,C)$-compatible and
    $A'$ is $(C,D)$-compatible, then $A'\circ A$ is
    $(B,D)$-compatible;
  \item\label{it:pairing} pairing: if $(A_i)_{i\in\iota}$ are
    $(B,C)$-compatible, then $\langle A_i\rangle_{i\in\iota}$ is
    $(B,C^\iota)$-compatible;
  \item\label{it:product} product: if $A$ is $(B,C)$-compatible and
    $A'$ is $(B',C')$-compatible, then $A\times A'$ is
    $(B{\times}B',C{\times}C')$-compatible; 
  \end{enumerate}
  Moreover, for an endofunctor $B \colon \C\to\C $, 
  \begin{enumerate}[(i)]\setcounter{enumi}{3}
  \item\label{it:id} the identity functor $\mathit{Id}\colon \C\to\C$
    is $B$-compatible;
  \item\label{it:constant} the constant functor to the carrier of any
    $B$-coalgebra is $B$-compatible, in particular the final one if it
    exists;
  \item\label{it:coproduct} the coproduct functor $\coprod\colon
    \C^\iota\to\C$ is $(B^\iota,B)$-compatible. %, assuming that $\C$ has coproducts of an ordinal $\iota$.
    \end{enumerate}
\end{proposition}

\section{Up-to Techniques in a Fibration}
\label{sec:upto:fibrations}

Throughout this section we fix a bifibration $p\colon \P\to\B$,
an endofunctor $F \colon \B \to \B$, a lifting $\lift{F}\colon
\P \to \P$ of $F$ and a coalgebra $\xi\colon X \to FX$.
Intuitively, the studied system 
% (i.e., the coalgebra $\xi\colon X \to FX$)
lives in the base category $\B$ while its properties live in $\P_X$, 
the fibre above $X$. We thus instantiate the category $\C$ from the
previous section with $\P_X$.

As explained in the Introduction~$(\dagger)$, we discuss proof
techniques for the properties modelled as final coalgebras of the
functor $\xi^*\circ\lift{F}_X\colon \P_X \to \P_X$, that we refer
hereafter as $\pred{\lift{F}}{\xi}$. In $\Rel$, when $\lift{F}$ is
the canonical lifting,  $\pred{\Rel(F)}{\xi}$-coalgebras are exactly the
Hermida-Jacobs bisimulations~\cite{Hermida97structuralinduction}.

To obtain sound techniques for $\pred{\lift{F}}{\xi}$, it suffices to
find $\pred{\lift{F}}{\xi}$-compatible endofunctors on $\P_X$. We
provide such functors by giving conditions on the lifting $\lift F$,
abstracting away from the coalgebra $\xi$ at hand.

\subsection{Compatibility of Behavioural Equivalence Closure}
\label{ssec:beh:compat}

The most basic technique is up to behavioural equivalence, a prime
example of which is Milner's up to bisimilarity~\cite{Milner89}, where
a relation $R$ is mapped into ${\sim}R{\sim}$.
If $f$ is the unique morphism from $\xi$ to a final $F$-coalgebra (assumed to exist), 
behavioural equivalence is the kernel  of $f$. This leads us to  consider the functor 
$$\Beh=f^*\circ\textstyle{\coprod}_f:\P_X\to\P_X.$$
For the fibrations $\Pred\to\Set$ and  $\Rel\to\Set$ the functor $\Beh$ maps a predicate, respectively a relation, to its closure under behavioural equivalence.
The compatibility of $\Beh$ is an instance of:
\begin{theorem}\label{theo:beh}
  Suppose that $(\lift{F}, F)$ is a fibration map. For any
  $F$-coalgebra morphism $f\colon (X,\xi)\to (Y,\zeta)$, the functor
  $f^*\circ \coprod_f$ is $\pred{\lift{F}}{\xi}$-compatible.
\end{theorem}
\begin{proof}[Proof sketch]
  We exhibit a natural transformation 
  \[
  f^*\circ \textstyle{\coprod}_f\circ(\xi^*\circ\lift{F})\Ra (\xi^*\circ\lift{F})\circ f^*\circ \textstyle{\coprod}_f
  \]
  obtained by pasting the 2-cells $(a),(b),(c),(d)$ in the following diagram:
  \begin{equation*}
    % \label{eq:beh-eq-comp}
    \xymatrix@R=1.5em{
      \P_X\ar[r]^-{\lift{F}}\ar@{=}[dd] & \P_{FX}\ar[r]^-{\xi^*}\ar[dr]_-{\coprod_{Ff}}\ar@{}[ddl]|(.5){\Downarrow(b)}  & \P_X\ar[r]^-{\coprod_f}\ar@{}[d]|(.5){\Downarrow(d)} & \P_Y\ar[r]^-{f^*}\ar@{<-}[dl]^-{\zeta^*}\ar@{}[ddr]|(.5){\Downarrow(c)} & \P_X\ar@{=}[dd] \\
      & & \P_{FY}\ar[dr]^-{(Ff)^*}\ar@{}[d]|(.5){\Downarrow(a)} & & \\
      \P_X\ar[r]^-{\coprod_f} & \P_Y\ar[r]^-{f^*}\ar[ur]^-{\lift{F}} & \P_X\ar[r]^-{\lift{F}} & \P_{FX}\ar[r]^-{\xi^*} & \P_X \\
    }
  \end{equation*}
  \begin{enumerate}[(a)]
  \item Since $(\lift{F}, F)$ is a fibration map we have that 
    \begin{equation*}
      % \label{eq:7}
      \lift{F}f^*=(Ff)^*\lift{F}
    \end{equation*}
  \item is a consequence of Lemma~\ref{eq:2-cell-lifting-F} in
    \appref{app:sec:upto:fibrations}.
    % Using the counit of the adjunction $\coprod_{Ff}\dashv(Ff)^*$
    % from~\eqref{eq:7} we obtain a 2-cell
    % $$ \textstyle{\coprod}_{Ff}\lift{F}f^*\Ra\lift{F}$$
    % Using further the unit of the adjunction $\coprod_{f}\dashv f^*$
    % we obtain the 2-cell $(b)$.
  \item is a natural isomorphism and comes from the fact that $f$ is a
    coalgebra map and the fibration is split.
  \item is obtained from $(c)$ using the counit of $\coprod_{f}\dashv
    f^*$ and the unit of $\coprod_{Ff}\dashv (Ff)^*$.
    % \begin{equation*}
    %   \xymatrix{
    %     \textstyle{\coprod}_f\xi^*\ar@{=>}[r] & \textstyle{\coprod}_f\xi^*(Ff)^*\textstyle{\coprod}_{Ff}\ar@{=}[d]& \\ 
    %     &\textstyle{\coprod}_ff^*\zeta^*\textstyle{\coprod}_{Ff} \ar@{=>}[r]&
    %     \zeta^*\textstyle{\coprod}_{Ff}
    %   }
    % \end{equation*}
  \end{enumerate}
  (Note that this proof decomposes into a proof that $\coprod_f$ is
  $(\pred{\lift F}\xi,\pred{\lift F}\zeta)$-compatible, by pasting (b) and (d), and a proof
  that $f^*$ is $(\pred{\lift F}\zeta,\pred{\lift F}\xi)$-compatible, by pasting (a) and (c).
  These two independent results can be composed by
  Proposition~\ref{prop:modularity}\eqref{it:compo} to obtain the
  theorem.)
\end{proof}

\begin{corollary}
  If  $F$ is a $\Set$-functor preserving weak pullbacks 
  then  the behavioural equivalence closure functor
  $\Beh$ is $\cpred{F}{\xi}$-compatible.
\end{corollary}
\begin{proof}
  $(\Rel(F),F)$ is a fibration map whenever $F$ preserves weak
  pullbacks (see e.g.~\cite{HughesJ04}).  
\end{proof}

From Theorem~\ref{theo:beh} we also derive the soundness of up-to
$\Beh$ for unary predicates: the \emph{monotone predicate liftings} used in coalgebraic modal
logic~\cite{DBLP:journals/cj/CirsteaKPSV11} are fibration
maps~\cite{Jacobs:coalg}, thus the hypothesis of Theorem~\ref{theo:beh} are satisfied.

\subsection{Compatibility of Equivalence Closure}
\label{ssec:trans:compat}

In this section we show that compatibility of equivalence closure can be 
 modularly derived from compatibility of reflexive, symmetric and transitive 
closures. 
For the latter it suffices to prove that relational composition is compatible. Composition of relations can be expressed in a fibrational
setting, by considering the category $\Rel \times_{\Set} \Rel$ obtained as a pullback of the fibration 
$\Rel\to\Set$ along itself:
\[
\xymatrix{
\Rel \times_{\Set} \Rel\ar[d]\pullbackcorner\ar[r] & \Rel\ar[d] \\
\Rel\ar[r] & \Set \\
}
\]
Then relational composition is a functor $\otimes \colon \Rel \times_{\Set} \Rel \to
\Rel$ mapping $R,S\subseteq X \times X$ to their composition. As we will see as a corollary of Proposition~\ref{prop:modular-compatibility-with-bisim}, when proving compatibility of relational composition  with respect to $\pred{\lift{F}}{\xi}$ we can abstract away from the coalgebra $\xi$ and simply use that $\otimes$ is a morphism of endofunctors from $\lift{F}^2$ to $\lift{F}$. 

Compatibility of symmetric and reflexive closures can be proved following the same principle. 
This leads us to consider for an arbitrary fibration $\P\to\B$ its $n$-fold product in the category
$\Fib(\B)$, denoted by $\P^{\times_{\B}n}\to \B$. The objects in
$\P^{\times_{\B}n}$ are tuples of objects in $\P$ belonging to the
same fibre. This product is computed fibrewise, that is,
$\P^{\times_{\B}n}_X{=}\P_X^n$.  For $n{=}0$ we have $\P^{0}{=}\B$.
% \footnote{This
%     makes sense, since $\id\colon \B\to\B$ is the final object in the
%     category of $\B$-indexed fibrations.}

Hereafter, we are interested in functors $G\colon \P^{\times_{\B}n}\to \P$
that are liftings of the identity functor on $\B$: for each $X$ in
$\B$ we have functors $G_X\colon \P_X^n\to\P_X$. Then relational composition is just an instance of $G$ for $n{=}2$.

\begin{proposition}\label{prop:modular-compatibility-with-bisim}
  Let $G\colon \P^{\times_{\B}n}\to \P$ be a lifting of the identity,
  with a natural transformation $G\lift{F}^n\Ra\lift{F}G$. Then $G_X$
  is $\pred{\lift{F}}{\xi}$-compatible.
\end{proposition}

We list now several applications of the proposition for the fibration
$\Rel\to\Set$.
\begin{enumerate}
\item[$(n{=}0)$] Let $\Ref\colon \Set\to\Rel$ be the functor mapping
  each set $X$ to $\Delta_X$, the identity relation on $X$.  $\Ref_X$
  is $\pred{\lift{F}}{\xi}$-compatible if 
  \begin{equation}
    \label{(*)}
    \Delta_{FX}\subseteq \lift{F}\Delta_X.\tag{$*$}
  \end{equation}
\item[$(n{=}1)$] Let $\Sym \colon \Rel\to\Rel$ be the functor mapping
  each relation $R\subseteq X^2$ to its converse $R^{-1}\subseteq
  X^2$. $\Sym_X$ is $\pred{\lift{F}}{\xi}$-compatible if for all
  relations $R\subseteq X^2$
  \begin{equation}
    \label{(**)}
    \lift F(R)^{-1} \subseteq \lift F(R^{-1})\tag{$**$} .    
  \end{equation}
\item[$(n{=}2)$] Let $\otimes\colon \Rel\times_\Set\Rel\to\Rel$ be the
  relational composition functor. Then $\otimes_X$ is
  $\pred{\lift{F}}{\xi}$-compatible if for all
  $R,S\subseteq X^2$ 
  \begin{equation}
    \label{(***)}
    \lift{F}R \otimes \lift{F}S \subseteq \lift{F}(R\otimes S)\tag{$*{*}*$}
  \end{equation}
  If moreover $T_1,T_2\colon\Rel_X\to\Rel_X$ are two
  $\pred{\lift{F}}{\xi}$-compatible functors, their pointwise composition $T_1\otimes
  T_2=\otimes_X\circ\langle T_1,T_2\rangle$ is
  $\pred{\lift{F}}{\xi}$-compatible by Proposition~\ref{prop:modularity}~(\ref{it:compo},\ref{it:pairing}).
\end{enumerate}
The transitive closure functor $\Tra$ is obtained from $\otimes$ in a
modular way: 
$$\Tra=\coprod_{i\geq 0} (-)^i\colon\P_X\to\P_X$$ where $(-)^0=\Id$ and
$(-)^{i+1}=\Id \otimes (-)^i$. 
Using
Proposition~\ref{prop:modularity} %~(\ref{it:compo},\ref{it:pairing},\ref{it:id},\ref{it:coproduct}),
we get

\begin{corollary}
  % Let $(X, \xi)$ be a coalgebra for a functor $F\colon \Set \to \Set$.
If $F$ is a $\Set$-functor  then the reflexive and symmetric closure functors $\Ref_X$ and $\Sym_X$
  are $\cpred{F}{\xi}$-compatible. Moreover, if $F$ preserves weak
  pullbacks, then the transitive closure functor $\Tra_X$ is
  $\cpred{F}{\xi}$-compatible.
\end{corollary}
\begin{proof}
  The above conditions $(*)$ and $(**)$ always hold for the canonical
  lifting $\lift{F}=\Rel(F)$; $(*{*}*)$ holds for $\Rel(F)$ when $F$
  preserves weak pullbacks.
\end{proof}

By compositionality (Proposition~\ref{prop:modularity}), one can then
deduce compatibility of the equivalence closure functor: this functor
can be defined as
%\begin{align*}
  $\Eqv \triangleq \Tra\circ(\Id + \Sym + \Ref)$, where $+$ denotes
  binary coproduct.
%\end{align*}

When $\pred{\lift{F}}{\xi}$ has a final coalgebra $S$, one can define
a ``self closure'' $\P_X$-endofunctor $\Slf=\widetilde{S}\otimes\Id\otimes
\widetilde{S}$, where $\widetilde{S}:\P_X\to\P_X$ is the constant to $S$
functor. Thanks to Proposition~\ref{prop:modularity}, the functor $\Slf$ is
$\pred{\lift{F}}{\xi}$-compatible whenever $(*{*}*)$ holds. 
When $F$ preserves weak pullbacks and $\lift F$ is instantiated to the canonical 
lifting $\Rel(F)$, $\Slf$ coincides with $\Beh$ since $S$ is just behavioural equivalence in this case.
If instead we consider the lifting % of $F$ ($F\times F$ actually)
that yields weak bisimilarity (to be defined in
Section~\ref{ssec:weak}), $\Slf$ corresponds to a technique called
``weak bisimulation up to weak bisimilarity'', while $\Beh$
corresponds to ``weak bisimulation up to (strong) bisimilarity''.

\subsection{Compatibility of Contextual Closure}
\label{ssec:context:compat}

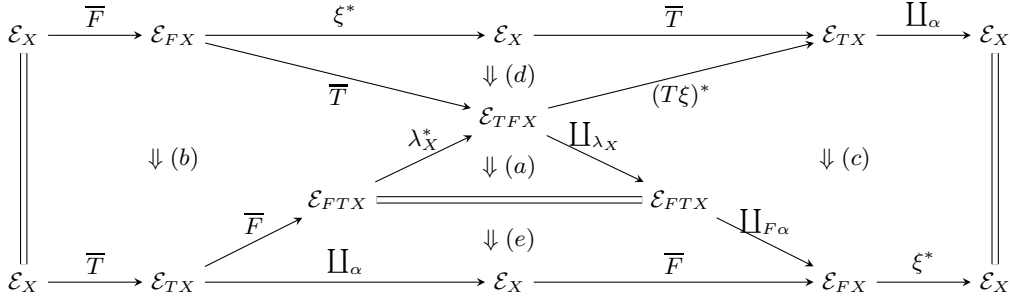
\begin{figure*}[t]
  \centering
  \begin{tikzpicture}
    \matrix (m) [matrix of math nodes,row sep=2em,column
    sep=4em,minimum width=2em] {
      \P_X & \P_{FX} &  & \P_X &  & \P_{TX} & \P_X \\
      &  & & \P_{TFX} & & & \\
      & {} &  \P_{FTX} & {} & \P_{FTX} & & \\
      \P_X & \P_{TX} & & \P_X & & \P_{FX} & \P_X \\}; \path[-stealth]
    (m-1-4) edge[color=white] node[fill=white]
    {$\textcolor{black}{\Downarrow(d)}$} (m-2-4) (m-1-1) edge
    [double,-,double distance=2pt] (m-4-1) edge node [above] {$
      \lift{F}$} (m-1-2) (m-1-2) edge node [above] {$\xi^*$} (m-1-4)
    edge node [below] {$\lift{T}$} (m-2-4) edge[color=white]
    node[fill=white] {$\textcolor{black}{\Downarrow(b)}$} (m-4-2) (m-1-4)
    edge node [above] {$\lift{T}$} (m-1-6) (m-1-6) edge node [above]
    {$\coprod_\alpha$} (m-1-7) edge[color=white] node[fill=white]
    {$\textcolor{black}{\Downarrow(c)}$} (m-4-6) (m-1-7) edge
    [double,-,double distance=2pt] (m-4-7)
    % (m-2-2) edge [double distance =1pt] (m-3-2)
    (m-2-4) edge node [below] {$(T\xi)^*$} (m-1-6) edge node [above]
    {$\coprod_{\lambda_X}$} (m-3-5) edge[color=white] node[fill=white]
    {$\textcolor{black}{\Downarrow(a)}$} (m-3-4) (m-3-3) edge node
    [above] {$\lambda_X^*$} (m-2-4) edge [double,-,double distance=2pt] 
    (m-3-5) (m-3-4) edge[color=white] node[fill=white]
    {$\textcolor{black}{\Downarrow(e)}$} (m-4-4) (m-3-5) edge node
    [above] {$\coprod_{F\alpha}$} (m-4-6) (m-4-1) edge node [above]
    {$\lift{T}$} (m-4-2) (m-4-2) edge node [above] {$\coprod_{\alpha}$}
    (m-4-4) edge node [above] {$\lift{F}$} (m-3-3) (m-4-4) edge node
    [above] {$\lift{F}$} (m-4-6) (m-4-6) edge node [above] {$\xi^*$}
    (m-4-7);
    % (m-2-1) edge node [below] {$=$} (m-2-2)
            
    % (m-1-2) edge node [right] {} (m-2-2)
    % % (m-2-1) edge [loop left] node [below] {$T$} (m-2-1.west);
    % (m-2-2) edge node [below] {$\Fcal$} (m-2-3) (m-1-2) edge node
    % [above] {$\Fcal_{\mathcal E\mathcal M}$} (m-1-3) (m-1-3) edge
    % node [above] {$U$} (m-1-4) edge node [below] {} (m-2-3) (m-2-3)
    % edge node [below] {$U$} (m-2-4) (m-1-4) edge node [below] {}
    % (m-2-4); \draw[vecArrow] (m-2-2) to (m-3-2); \draw[vecArrow]
    % (m-2-4) to (m-3-4);
    % % \draw[innerWhite] (m-2-2) to (m-3-2);
    % \draw[-,shorten <=5pt,shorten >=5pt,line width=3.8pt] (m-1-1) to
    % (m-4-1); \draw[-,shorten <=4pt,shorten >=4pt,line width=3pt,
    % color=white] (m-1-1) to (m-4-1);
  \end{tikzpicture}
  \nocaptionrule
  \caption{Compatibility of contextual closure in a fibration}
  \label{eq:compatibility}
\end{figure*}

For defining contextual closure, we assume that the state space of
the coalgebra  is equipped with an algebraic structure.
More precisely, we fix a bialgebra for a distributive law $\lambda\colon TF\Ra FT$, that is, a triple $(X,\alpha,\xi)$, where $\alpha:TX\to X$ is a an algebra and $\xi \colon X \to FX$ is a coalgebra  such that the next diagram commutes:
\begin{align*}\qquad\qquad
\begin{tikzpicture}[
font=\sffamily\footnotesize,
sink/.style={color=black, draw,thick,rounded corners,inner sep=.3cm,fill=green!20}
]
    \matrix (m) [matrix of nodes,row sep=3em,column
    sep=4em,minimum width=2em,ampersand replacement=\&] {
      \node (m-1-1) {$TX$}; \&  
      \node (m-1-2) {$X$}; \&  
      \node (m-1-3) {$FX $}; \&  
     \\      
     \node (m-2-1) {$TFX$}; \&  
      \node (m-2-2) {}; \&  
      \node (m-2-3) {$FTX$}; \&  
\\}; 
      \path[-stealth]
     (m-1-1) edge node [left] {$T\xi$} (m-2-1) 
     edge node [above] {$\alpha$} (m-1-2) 
     (m-1-2) edge node [above] {$\xi$} (m-1-3)
(m-2-1) edge node [below] {$\lambda_X$} (m-2-3)
(m-2-3) edge node [right] {$F\alpha$} (m-1-3);
  \end{tikzpicture}
\end{align*}

\begin{theorem}
  \label{thm:big-2-cell}
  Let $\lift{T},\lift{F}\colon \P\to\P$ be liftings of $T$ and $F$. If
  $\lift{\lambda} \colon \lift{T}\,\lift{F}\Ra\lift{F}\,\lift{T}$ is a
  natural transformation sitting above $\lambda$, then
  $\coprod_\alpha\circ\,\lift{T}$ is $\pred{\lift{F}}{\xi}$-compatible.
\end{theorem}

\begin{proof}[Proof sketch]
  We exhibit a natural transformation
  $$(\textstyle{\coprod}_\alpha\circ \lift{T})\circ(\xi^*\circ\lift{F})
  \Rightarrow (\xi^*\circ \lift{F})\circ(\textstyle{\coprod}_\alpha\circ \lift{T})\enspace.$$
  This is achieved in Figure~\ref{eq:compatibility} by pasting five
  natural transformations, obtained as follows:
  \begin{enumerate}[(a)]
  \item is the counit of the adjunction
    $\coprod_{\lambda_X}\dashv\lambda_X^*$.
  \item comes from $\lift{\lambda}$ being a lifting of
    $\lambda$. %, see Lemma~\ref{eq:2-cell-distributivity}.
  \item comes from the bialgebra condition, the fibration being split,
    and the units and counits of the adjunctions
    $\coprod_{\alpha}\dashv\alpha^*$,
    $\coprod_{F\alpha}\dashv(F\alpha)^*$, and 
    $\coprod_{\lambda_X}\dashv\lambda_X^*$. % See Lemma~\ref{eq:2-cell-bialgebra}.
  \item arises since $\lift{T}$ is a lifting of $T$, using the
    universal property of the Cartesian lifting
    $(T\xi)^*$. % , see Lemma~\ref{eq:2-cell-lifting-T}.
  \item comes from $\lift{F}$ being a lifting of $F$, combined with
    the unit and counit of the adjunction
    $\coprod_{\alpha}\dashv\alpha^*$.
    % , see Lemma~\ref{eq:2-cell-lifting-F}. 
  \end{enumerate}
  (Note that like for Theorem~\ref{theo:beh}, this proof actually
  decomposes into a proof that $\lift{T}$ is
  $(\pred{\lift{F}}{\xi},(T\xi)^*\circ\lambda_X^*)$-compatible, and a
  proof that $\coprod_\alpha$ is
  $((T\xi)^*\circ\lambda_X^*,\pred{\lift{F}}{\xi})$-compatible.)
\end{proof}

When the fibration at issue is $\mathsf{Rel}\to\Set$ and $\lift{T}$ is
the canonical lifting $\Rel(T)$, one can easily check that
$\coprod_\alpha\!\circ\, \Rel(T)$ applied to a relation $R$ gives exactly
its \emph{contextual closure} as described in~\cite{RotBBPRS}. For
this reason, we abbreviate $\coprod_\alpha\!\circ\, \Rel(T)$ to $\Con$.
When moreover $\lift{F}$ is the canonical lifting $\Rel(F)$, we get:
\begin{corollary}[{\cite[Theorem~4]{RotBBPRS}}]\label{cor:context}
  If $F,T$ are $\Set$-functors and $(X, \alpha, \xi)$ is a bialgebra for $\lambda \colon T F
  \Rightarrow F T$.  The contextual closure functor $\Con$
  is $\cpred{F}{\xi}$-compatible.
\end{corollary}
\begin{proof}
  The canonical lifting $\Rel(-)$ is a 2-functor~\cite[Exercise 4.4.6]{Jacobs:coalg}. Therefore 
  $\lift\lambda=\Rel(\lambda)$ fulfils the assumption of  Theorem~\ref{thm:big-2-cell}.
\end{proof}
Our interest in Theorem \ref{thm:big-2-cell} is not restricted to
prove compatibility of up to $\Con$. By taking non canonical liftings
of $T$, one derives novel and effective up-to techniques, such as the
\emph{monotone contextual closure} and the \emph{left-contextual
  closure} defined in Sections \ref{ssec:weighted} and
\ref{ssec:divergence}.
In order to apply Theorem \ref{thm:big-2-cell} for situations when
either $\lift{T}$ or $\lift{F}$ is not the canonical relation lifting,
one has to exhibit a $\lift{\lambda}$ sitting above $\lambda$. In
$\Rel$, such a $\lift{\lambda}$ exists if and only if for all
relations $R\subseteq X^2$, the restriction of $\lambda_X \times
\lambda_X$ to $\lift{T}\,\lift{F}R$ corestricts to
$\lift{F}\,\lift{T}R$.
A similar condition has to be checked for $\mathsf{Pred}\to\Set$.

\subsection{Abstract GSOS}
\label{ssec:abstract-GSOS}

For several applications, it is convenient to consider natural
transformations of a slightly different type $\lambda \colon T(F
\times \Id) \Rightarrow F\T$, where $\T$ is the free monad over
$T$. These are called \emph{abstract GSOS specifications} since, as
shown in \cite{DBLP:conf/lics/TuriP97}, they generalise GSOS rules to
any behaviour endofunctor $F$.  Each such $\lambda$ induces a
distributive law $\lambda^{\dagger}\colon \T(F \times \Id) \Rightarrow
(F \times \Id)\T$ of the monad $\T$ over the copointed functor $F
\times \Id$, whose bialgebras are the objects of our interest (see
\appref{app:GSOS}). In order to prove compatibility via Theorem
\ref{thm:big-2-cell}, one should exhibit a $\lift{\lambda^{\dagger}}$
sitting above $\lambda^{\dagger}$. The following lemma simplifies such
a task.
%\dam: maybe we should point out that this is a generalisation of
%\cite[Prop.~6.4.10]{PS12}}
\begin{lemma}\label{lm:gsos}
  Let $\lift{\lambda}\colon \lift{T}(\lift{F}\times\lift{\Id})\Rightarrow
  \lift{F}\,\lift{\T}$ be a natural transformation sitting above
  $\lambda\colon T(F\times\Id)\Rightarrow F\T$.  Then there exists a
  $\lift{\lambda^\dagger}
  \colon\lift{\T}(\lift{F}\times\lift{\Id})\Rightarrow
  (\lift{F}\times\lift{\Id})\lift{\T} $ sitting above
  $\lambda^\dagger\colon\T(F\times\Id)\Rightarrow (F\times\Id)\T $.
\end{lemma}
For a bialgebra $(X,\alpha, \langle \xi, id \rangle)$, the existence
of $\lift{\lambda^\dagger}$ ensures, via Theorem \ref{thm:big-2-cell},
compatibility w.r.t.\ $\pred{(\lift{F}\times \Id)}{\langle
  \xi,id\rangle }$, which is not exactly
$\pred{\lift{F}}{\xi}$. However, this difference is harmless in poset
fibrations: coalgebras for the two functors coincide, and for any
pointed functor $A$ compatible with $\pred{(\lift{F}\times
  \Id)}{\langle \xi,id\rangle }$, every
$\pred{\lift{F}}{\xi}$-invariant up to $A$ is also an
$\pred{(\lift{F}\times \Id)}{\langle \xi,id\rangle }$-invariant up to
$A$. 

\section{Examples}
\label{sec:examples}

\subsection{Inclusion of weighted automata}\label{ssec:weighted}

\newcommand{\sem}{\mathbb{S}}

To illustrate how to instantiate the above framework, we consider
weighted automata.  We first give a short description of their
coalgebraic treatment~\cite{Bonchi12}.
For a semiring $\sem$ and a set $X$, we denote by $\sem^X_\omega$ the
set of functions $f \colon X \rightarrow \sem$ with finite support.
% ,
% that is, such that $f(x) \neq 0$ for finitely many $x$. 
These
 functions can be thought of as linear combinations $\sum_{x\in
   X}f(x)\cdot x$, and in fact
$\sem^-_\omega\colon \Set \to \Set$ is the monad
sending each set $X$ to the free semi-module generated by $X$.

A \emph{weighted automaton} over a semiring $\sem$ with alphabet $A$
is a pair $(X,\langle o,t\rangle)$, where $X$ is a set of states,
$o\colon X \to \sem$ is an output function associating to each state
its output weight and $t\colon X \to (\sem^X_\omega)^A$ is a
weighted transition relation. Denoting by $F$ the functor $\sem \times
(-)^A$, weighted automata are thus coalgebras for the composite
functor $F\sem^-_\omega$. By the generalised powerset
construction~\cite{SilvaBBR10}, they induce bialgebras for the functor
$F$, the monad $\sem^-_\omega$, and the distributive law $\lambda
\colon \sem^-_\omega F \Rightarrow F\sem^-_\omega$ given for all sets
$X$ by $\lambda_X(\sum r_i(s_i,\varphi_i)) = \langle \sum r_i s_i,
\lambda a . \sum r_i \varphi_i(a) \rangle$.  Indeed every $(X,\langle
o,t\rangle)$ induces a bialgebra $(\sem^X_\omega, \mu, \langle
o^{\sharp},t^{\sharp} \rangle )$ where $\mu$ is the multiplication of
$\sem^-_\omega$ and 
% $\langle o^{\sharp},t^{\sharp} \rangle= (\sem
% \times \mu^A) \circ \lambda \circ (\sem^{ \langle o,t
%   \rangle})$. Intuitively,
$\langle o^{\sharp},t^{\sharp} \rangle\colon \sem^X_\omega \to \sem
\times (\sem^X_\omega)^A $ is the linear extension of $\langle
o,t\rangle$, defined as 
%$(\sem \times \mu^A) \circ \lambda \circ(\sem^{ \langle o,t \rangle}_{\omega})$.
$(F \mu) \circ \lambda \circ(\sem^{ \langle o,t \rangle}_{\omega})$.

For a concrete example we take the semiring $\R^+$ of positive real
numbers. A weighted automaton is depicted on the left below: arrows
$x\tr{a,r}y$ mean that $t(x)(a)(y)=r$ and arrows $x
\stackrel{r}{\Rightarrow}$ mean that $o(x)=r$.
\vspace{-.7em}
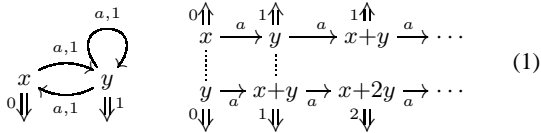
\begin{align}
\label{ex:weighted-aut}
\begin{split}
  \quad\nfa{\xymatrix @C=.8em @R=1em {\\\\ %
           \state{x} \ar@{=>}[d]_{0} \ar@/^/[rr]^{a,1}  && %
          \state{y} \ar@{=>}[d]^{1} \ar@(ul,ur)^{a,1}
          \ar@/^/[ll]^{a,1}\\ & &  }}
      &\qquad
      \dfa{\xymatrix @C=1em @R=.8em { 
       &&&\\%
          \state{x} \ar@{=>}[u]^{0} \ar[r]^{a} & \state{y} \ar@{=>}[u]^{1} \ar[r]^{a} &\state{x{+}y} \ar@{=>}[u]^{1} \ar[r]^{a} &\cdots \\%
          \state{y} \ar@{.}[u] \ar@{=>}[d]_{0} \ar[r]_(0.4){a} &
          \state{x{+}y} \ar@{.}[u] \ar@{=>}[d]_{1} \ar[r]_(0.4){a}
          &\state{x{+}2y} %\ar@{.}[u] 
          \ar@{=>}[d]_{2} \ar[r]_(0.5){a} & \cdots \\ &&& }}
\end{split}
\end{align}
On the right is depicted (part of) the corresponding bialgebra: states
are elements of $(\R^+)_\omega^X$ (hereafter denoted by $v,w$), arrows
$v\tr{a}w$ mean that $t^{\sharp}(v)(a)=w$ and arrows $v
\stackrel{r}{\Rightarrow}$ mean that $o^{\sharp}(v)=r$. %The language $\bb{x}$
%accepted by $x$ maps the word $a^n$ into the $n$-th
%Fibonacci number.

\smallskip

Whenever $\sem$ carries a partial order $\leq$, one can take the following lifting
$\bar{F} \colon \Rel \rightarrow \Rel$ of $F$
%\colon \Set \to \Set
defined for $R\subseteq X^2$ by:
$$  \{((r,\varphi),(s,\psi)) \mid r \leq s\ \land\  \forall
  a. \varphi(a) \mathrel R \psi(a)\}\subseteq (FX)^2.
%FX {\times} FX
$$

\newcommand{\W}{\pred{\lift{F}}{\langle o^{\sharp},t^{\sharp}\rangle}}
Then the functor $\W=\langle o^{\sharp},t^{\sharp} \rangle^* \circ
\bar{F}\colon \Rel_X \to \Rel_X$ maps a relation $R\subseteq X^2$ into
\begin{align*}
  \{(x,y) \mid o^{\sharp}(x)\leq o^{\sharp}(y) \land \forall a.
  t^{\sharp}(x)(a) \mathrel R t^{\sharp}(y)(a) \}\enspace.
\end{align*}

The carrier of a final $\W$-coalgebra is a relation, denoted by $\precsim$, which we call \emph{inclusion}:
when $\sem$ is the Boolean semiring, it coincides with language inclusion of non-deterministic automata.

For any two $v,w\in\sem^X_\omega$, one can prove that $v\precsim w$ by
exhibiting a $\W$-invariant relating them. These invariants are
usually infinite, since there are infinitely many reachable states in
a bialgebra $\sem^X_\omega$, even for finite $X$. This is the case
when trying to check $x\precsim y$ in~\eqref{ex:weighted-aut}: we should
relate infinitely many reachable states.

In order to obtain finite proofs, we exploit the algebraic structure
of bialgebras and employ an up to context technique.
%
%We can use for that the canonical lifting of the monad
To this end, we use the canonical lifting of the monad
$\sem^-_\omega$, defined for all $R \subseteq X^2$ as
\begin{align*}
  \Rel(\sem^-_\omega)(R) &= \left\{\left(\sum r_ix_i,~\sum
      r_iy_i\right) \mid x_i \mathrel R y_i\right\}
\end{align*}

We prove that the endofunctor $\Con=\coprod_\mu \circ
\Rel(\sem^-_\omega) $ is $\W$-compatible by
Theorem~\ref{thm:big-2-cell}: it suffices to check that for any
relation $R$ on $X$, the restriction of
$\lambda_X{\times}\lambda_X$ to $\Rel(\sem^-_\omega)\bar{F} (R)$
corestricts to $\bar{F}\Rel(\sem^-_\omega)(R)$. This is the case when
for all $n_1, m_1, n_2, m_2 \in \sem$ such that $n_1 \leq m_1$ and
$n_2 \leq m_2$, we have (a) $n_1 + n_2 \leq m_1 + m_2$ and (b) $n_1
\cdot n_2 \leq m_1 \cdot m_2$.  These two conditions are satisfied,
e.g., in the Boolean semiring or in $\R^+$ and thus, in
these cases, we can prove
inclusion of automata using $\W$-invariants up to $\Con$. For example, in~\eqref{ex:weighted-aut}, the
relation $R=\{(x,y),(y,x{+}y)\}$ is a $\W$-invariant up to $\Con$ (to check this, just observe that $(x{+}y, x{+}2y)\in \Con(R)$). This finite
relation thus proves $x \precsim y$.

\smallskip

Unfortunately, condition~(b) fails for the semiring $\R$ of (all) real
numbers. Nevertheless, our framework allows us to define another up-to technique, 
which we call ``up to
\emph{monotone} contextual closure''. It is obtained by composing $\coprod_\mu$ and 
a non-canonical lifting of $\R^-_\omega$:
\begin{align*}
  \bar{\R^-_\omega}(R) &= \left\{\left(\sum r_ix_i,\sum r_iy_i\right)
    \mid
    \begin{array}{r}
      r_i \geq 0 \Rightarrow x_i \mathrel R y_i\\  
      r_i < 0 \Rightarrow y_i \mathrel R x_i
    \end{array}\right\}
\end{align*}
The restriction of $\lambda_X \times
\lambda_X$ to $\bar{\R^-_\omega}\bar{F} (R)$ corestricts to
$\bar{F}\bar{\R^-_\omega}(R)$.  Therefore, by
Theorem~\ref{thm:big-2-cell}, the monotone contextual closure is
$\W$-compatible.
%

%%%%%%%%%%%%%%%%%%%%%%%%%%%%%%%%%%%%%%%%%%%%%%%%%%%%%%%%%%%%%%%%%%%%%%%%

%
%
\subsection{Divergence of processes}\label{ssec:divergence}

%
% We now show how up-to techniques can be useful to work with unary
% predicates, using the fibration $p:\Pred\to\Set$ from Example
% \ref{ex:pred-fib}.
Up-to techniques can be instrumental in proving unary predicates. We
take the fibration $\Pred\to\Set$ and we focus on the
\emph{divergence} predicate $\nu u. \diamondtau u$ defined on
LTSs. The latter are coalgebras $\xi \colon X \to F(X)$ for the
$\Set$-functor $FX=\pow(L\times X)$, where
$L=\{a,\bar{a},b,\bar{b}\dots, \tau\}$ is a set of labels containing a
special symbol $\tau$ and $\pow$ is the \emph{finite} powerset
functor.
We lift $F$ to $\lift{F}^{\diamondtau}\colon \Pred \to \Pred$,
defined for all sets $X$ as
%$\lift{FX}^{\diamondtau}, \lift{FX}^{\diamondtau} \colon \Pred_X \to \Pred_{FX}$
$$\lift{F}^{\diamondtau}_X(P\subseteq X) = \{ S \in  FX \mid  \exists (\tau,x)\in S, x\in P \}.$$
The final $\pred{\lift{F}^{\diamondtau}}{\xi}$-coalgebra consists
precisely of all the states in $X$ satisfying $\nu u. \diamondtau
u$.  Hence, to prove that a state $p$ diverges, it suffices to
exhibit an $\pred{\lift{F}^{\diamondtau}}{\xi}$-invariant containing
$p$.

When the LTS is specified by some process algebra, such invariants
might be infinite. Suppose for instance that we have a parallel
operator defined by the following GSOS rules and their symmetric
counterparts:
\begin{align*}
  \frac{x\tr{l}x'}{x|y\tr{l}x'|y} &&\frac{x\tr{a}x' \quad
    y\tr{\bar{a}}y'}{x|y\tr{\tau}x'|y'}\enspace.
\end{align*}
Consider the processes $p\tr{a}p|p$ and $q\tr{\bar a}q$.  To prove
that $p|q$ diverges, any invariant should include all the states that
are on the infinite path $p|q \tr{\tau} (p|p)|q \tr{\tau}\dots$.

Instead, an intuitive proof would go as follows: assuming 
that $p|q$ diverges one has to prove that the $\tau$ successor $(p|p)|q$ also
diverges. Rather than looking further for the $\tau$-successors of
$(p|p)|q$, observe that
\begin{enumerate}[(a)]
\item since $p|q$ diverges by hypothesis, then also $(p|q)|p$
  diverges, and
\item since $(p|q)|p$ is bisimilar (i.e., behavioural equivalent) to
  $(p|p)|q$, then also $(p|p)|q$ diverges.
\end{enumerate}
Formally, (b) corresponds to using the functor $\Beh$ from
Section~\ref{ssec:beh:compat}. For (a) we define the \emph{left
  contextual closure} functor as $\Con^{\ell}(P{\subseteq}X)=\{(\ldots(x|y_1)|\ldots) |y_n
\mid x\in P,~y_i\in X\}$. Indeed, it is easy to see that $P=\{p|q\}$ is
an $\pred{\lift{F}^{\diamondtau}}{\xi}$-invariant up to $\Beh\circ
\Con^{\ell}$, i.e, $P\subseteq \pred{\lift{F}^{\diamondtau}}{\xi}
\circ \Beh\circ \Con^{\ell}(P)$.

In order to prove soundness of this ``up to behavioural equivalence
and left contextual closure'', we show compatibility of $\Beh$ and
$\Con^{\ell}$ separately.
For the former, we note that $\lift{F}^{\diamondtau}$ is defined
exactly as in coalgebraic modal logic
\cite{DBLP:journals/cj/CirsteaKPSV11,HasuoCKJ:coindFib} and thus
$(\lift{F}^{\diamondtau}, F)$ is a fibration map:
Theorem~\ref{theo:beh} applies.
The functor $\Con^{\ell}$ is defined just as $\Con$, but instead of
the canonical lifting of the endofunctor for binary operations
$T(X)=X{\times}X$ we use the predicate lifting
$\lift{T}(P\subseteq X) = P{\times} X\subseteq TX$.
The conditions of Lemma \ref{lm:gsos} are met for the distributive law
given by the above GSOS rules (see~\appref{app:divergence}).
The functor $\Con^{\ell}$ can be seen to be the composition
$\textstyle{\coprod}_\mu\circ\lift{\mathbb{T}}$ where
$\lift{\mathbb{T}}$ is the free monad on $\lift{T}$ and $\mu$ is the
multiplication of $\mathbb{T}$. We can thus apply
Theorem~\ref{thm:big-2-cell} and obtain its compatibility.

\subsection{Equivalence of nominal automata}
\label{ssec:nominal-automata}

Nominal automata and variants~\cite{BojanczykKL11} have been
considered as a means of studying languages over infinite alphabets,
but also for the operational semantics of process
calculi~\cite{MontanariP05}. We refer the reader to~\cite{Pitts-book}
for background on the category $\Nom$ of nominal sets. These are sets equipped with actions of the group of permutations on a countable
set $\At$ of names, satisfying an additional finite support condition.

Consider the nominal automaton below. The part 
%which is 
reachable from state $*$ corresponds to~\cite[Example~I.1]{BojanczykKLT13}. 
\vspace{-.5em}
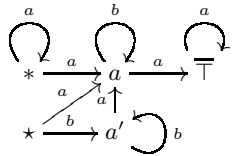
\begin{equation*}
  \nfa{\xymatrix @C=.8em @R=1em { %
      \state{{*}}\ar@(ul,ur)^-{a}  \ar[rr]^-{a}&  & %
      \state{a}  \ar@(ul,ur)^-{b}\ar[rr]^-{a} & & \state{\bar{\top}}\ar@(ul,ur)^{a}\\
      \state{\star}\ar[rru]^-{a}\ar[rr]^-{b} & & \state{a'}\ar@(rd,ru)_-{b}\ar[u]^-{a} & &
    }}
\end{equation*}
It is important to specify how to read this drawing: the represented
nominal automaton has as state space the orbit-finite nominal set
$\{*\}+\{\star\}+\At+\At'+\{\top\}$, where $\At'$ is a copy of
$\At$. It suffices in this case to give only one representative of
each of the five orbits: we span all the transitions and states of the
automaton by applying all possible finite permutations to those
explicitly written. For example, the transition $a\tr{c} a$ is
obtained from $a\tr{b}a$ by applying the transposition $(b\ c)$ to the
latter.

With this semantics in mind, one can see that the state $*$ accepts
the language of words in the alphabet $\At$ where some letter appears
twice: it reads a word in $\At$, then it nondeterministically guesses
that the next letter will appear a second time and verifies that this
is indeed the case.  The state $\star$ accepts the same language, in a
different way: it reads a first letter, then guesses if this letter
will be read again, or, if a distinct letter---nondeterministically
chosen---will appear twice.

Formally, nominal automata are $F\Pf$-coalgebras $\langle o, t\rangle$ 
%on $\{*\}+\{\star\}+\At+\At'+\{\top\}$, 
where $F\colon \Nom\to\Nom$ is given by $FX=2\times X^\At$ and the monad $\Pf$ is the
finitary version of the power object functor in the category of
nominal sets (mapping a nominal set to its finitely-supported
orbit-finite subsets). In our example, $o(a)=0$ and $t(a)$ is the following map:
\begin{equation*}
  t(a) = \left\{
  \begin{array}{lc}
    b\mapsto \{a\} & b\# a\\
    a\mapsto \{\top\} & \\
\end{array}
\right.
\end{equation*}

By the generalised powerset construction~\cite{SilvaBBR10}, $\langle o,t \rangle$ induces a deterministic nominal automata, which is a bialgebra on $\Pf(X)$ with the algebraic structure given by union. To prove that $*$ and $\star$ accept the same language, we should play
the bisimulation game in the determinisation of the
automaton. However, the latter has \emph{infinitely} many orbits and a rather
complicated structure. A bisimulation constructed like this will thus
have infinitely many orbits.
Instead, we can show that the orbit-finite relation spanned by the
four pairs
$$(\{*\},\{\star\}),~ (\{a\},\{a,a'\}),~ (\{\top\},\{ a,\top\}),~ (\{*\},\A') $$
is a bisimulation up to congruence (w.r.t.\ union).

The soundness of this technique is established in
\appref{ssec:nom-aut-ex} using the fibration $\Rel(\Nom)\to\Nom$ of
equivariant relations.
We derive the compatibility of contextual closure using
Theorem~\ref{thm:big-2-cell}, and compatibility of the transitive,
symmetric, and reflexive closures using
Proposition~\ref{prop:modular-compatibility-with-bisim}.
Compatibility of congruence closure follows from
Proposition~\ref{prop:modularity}\eqref{it:compo}.

\section{Compositional Predicates}\label{sec:compositional}

In this section we consider a structured way of defining coinductive
predicates, by composing lifted functors.  Assume a fibration
$p\colon \P\to\B$ and a functor $\otimes \colon \P\times_\B\P\to\P$. Given
two liftings $\bar{F_1},\bar{F_2} \colon \P \to \P$ of the same
endofunctor $F$ on $\B$, one can then define a \emph{composite}
lifting $\otimes \circ \langle \bar{F_1}, \bar{F_2} \rangle$, which we
denote by $\bar{F_1} \otimes \bar{F_2}$.  We will instantiate this to
the fibration $\Rel \rightarrow \Set$ with relational composition for
$\otimes$, to define simulation and weak bisimulation as coinductive
predicates.

One advantage of this approach is that the compatibility of
up-to-context can be proved in a modular way. %compositionally.

\newcommand{\thmcomplambda}{
  Let $\lift{T}$ be a lifting of $T$ having a
  $\gamma \colon \lift{T}\otimes\Ra\otimes\lift{T}^2$ above $\Id\colon T\Ra T$.
 Let both $\lift{F_1}$ and
  $\lift{F_2}$ be liftings of $F$.  If $\lambda_1\colon
  \lift{T}\,\lift{F_1} \Rightarrow \lift{F_1}\,\lift{T}$ and
  $\lambda_2 \colon \lift{T}\,\lift{F_2} \Rightarrow
  \lift{F_2}\,\lift{T}$ sit above the same $\lambda \colon TF
  \Rightarrow FT$, then there
  exists $\lift{\lambda} \colon \lift{T}(\lift{F_1} \otimes
  \lift{F_2}) \Rightarrow (\lift{F_1} \otimes \lift{F_2}) \lift{T}$
  above $\lambda$.
}
\begin{theorem}\label{thm:comp-lambda}
\thmcomplambda
\end{theorem}
Notice that the canonical lifting $\Rel(T)$ always satisfies the first
hypothesis of the theorem when $\otimes$ is relational composition.

\subsection{Simulation up-to}
\label{ssec:compositional:simulation}

We recall simulations for coalgebras as introduced in~\cite{HughesJ04}.
An endofunctor $F$ on $\Set$ is said to be \emph{ordered} if it factors
through the forgetful functor from $\Pre$ (the category of preorders)
to $\Set$: this means that for every $X$, $FX$ is equipped with a
preorder $\sqsubseteq_{FX}$. An ordered functor gives rise to a
\emph{constant} relation lifting $\lift{\sqsubseteq}$ of $F$ defined
as $\lift{\sqsubseteq}(R \subseteq X^2) = {\sqsubseteq_{FX}}$. Then the
\emph{lax relation lifting} $\laxlift{F}$ is defined as
\begin{align*}
  \laxlift{F} &= {\lift\sqsubseteq} \otimes \Rel(F) \otimes
  {\lift\sqsubseteq}
\end{align*}
where $\otimes$ is relational composition. For a coalgebra $\xi \colon
X \to FX$, the coalgebras for the endofunctor $\xi^* \circ
\laxlift{F}_X$---which we denote as $\pred{\laxlift{F}}{\xi}$---are
called \emph{simulations}; the final one is called \emph{similarity}.
We list two examples of ordered functors and their associated notion
of simulations, and refer to~\cite{HughesJ04} for many more.

\begin{example}\label{ex:ordered-functors}
  For weighted automata on a semiring $\sem$ equipped with a partial
  order $\leq$, the functor $FX=\sem \times X^A$ is ordered with
  $\sqsubseteq_{FX}$ defined as $(s,\phi) \sqsubseteq_{FX} (r,\psi)$
  iff $s\leq r$ and $\phi=\psi$. It is immediate to see that
  $\laxlift{F}$ coincides with the lifting $\lift{F}$
  defined in Section~\ref{ssec:weighted}.

  For LTSs, the functor $FX=\pow(A \times X)$ is ordered with subset
  inclusion $\subseteq$. In this case a simulation is a relation
  $R\subseteq X^2$ such that for all $(x,y) \in R$: if $x
  \xrightarrow{a} x'$ then there exists $y'$ such that $x'
  \xrightarrow{a} y'$ and $x' R y'$.
\end{example}

An ordered functor $F$ is called \emph{stable} if $(\laxlift{F},F)$ is a
fibration map~\cite{HughesJ04}. Since polynomial functors are stable, as well as the one for
LTSs~\cite{HughesJ04}, the following results hold for the coalgebras in
Example~\ref{ex:ordered-functors}.
\begin{proposition}
  If $F$ is a stable ordered functor, then $\Beh$, $\Slf$, and $\Tra$
  are $\pred{\laxlift{F}}{\xi}$-compatible.
\end{proposition}
\begin{proof}
  Compatibility of $\Beh$ comes from
  Theorem~\ref{theo:beh}. Compatibility of $\Slf$ and $\Tra$ comes
  from Proposition~\ref{prop:modular-compatibility-with-bisim}: stable
  functors satisfy
  $(*{*}*)$~\cite[Lemma~5.3]{HughesJ04}. 
\end{proof}

We proceed to consider the compatibility of up to context, for which
we assume an abstract GSOS specification $\lambda \colon T
(F{\times}\Id) \Rightarrow F\T$. By Theorem~\ref{thm:comp-lambda}, proving compatibility w.r.t.\
$\pred{\laxlift{F}}{\xi}$ is reduced to proving compatibility w.r.t. its components $\Rel(F)$ and $\lift{\sqsubseteq}$. 
For the former,
compatibility comes immediately from the proof of Corollary~\ref{cor:context}.
For
the latter, we need to assume that the abstract GSOS specification is
\emph{monotone}, i.e, such that for any set $X$, the restriction of
$\lambda_X \times \lambda_X$ to $\Rel(T) (\sqsubseteq_{FX}\times
\Delta_X)$ corestricts to $\sqsubseteq_{F\T X}$.  If $T$ is a
polynomial functor representing a signature, then this means that for
any operator $\sigma$ (of arity $n$) we have
$$
\frac{b_1 \sqsubseteq_{FX} c_1 \quad\ldots\quad b_n \sqsubseteq_{FX} c_n}
 {\lambda_X(\sigma(\mathbf{b,x})) \sqsubseteq_{F \T X} \lambda_X(\sigma(\mathbf{c,x}))}
$$
where $\mathbf{b,x} = (b_1,x_1), \ldots, (b_n,x_n)$ with $x_i\in X$
and similarly for $\mathbf{c,x}$.  If $\sqsubseteq$ is the order on
the functor for LTSs, monotonicity corresponds to the \emph{positive
  GSOS} format~\cite{MS10} which, as expected, is
GSOS~\cite{BloomCT88:GSOS} without negative premises.
Monotonicity turns out to be precisely the condition needed to apply 
Lemma~\ref{lm:gsos}, yielding
\newcommand{\propsim}{Let $\lambda$ be a monotone abstract GSOS specification and $(X, \alpha, \langle \xi, id\rangle )$ be a $\lambda^{\dagger}$-bialgebra.
Then $\Con$ is $\pred{(\laxlift{F}\times \Id)}{\langle \xi, id\rangle}$-compatible. 
}  
\begin{proposition}\label{prop:sim}
  \propsim
\end{proposition}

\subsection{Weak bisimulation-up-to}\label{ssec:weak}
A \emph{weak bisimulation} is a relation $R \subseteq X^2$ on the
states of an LTS such that for every pair $(x,y) \in R$: (1) if $x
\xrightarrow{l} x'$ then $y \stackrel{l}{\Rightarrow} y'$ with
$(x',y')\in R$ and (2) if $y \xrightarrow{l} y'$ then $x
\stackrel{l}{\Rightarrow} x'$ with $(x',y')\in R$.  Here $\rightarrow$
and $\Ra$ are two LTSs, i.e., coalgebras for the functor $FX = \pow(L
{\times} X)$, and $\Ra$ is the \emph{saturation}~\cite{Milner89} of
$\to$.
Weak bisimilarity can alternatively be reduced to strong bisimilarity
on $\Rightarrow$, but the associated proof method is rather
tedious. To remain faithful to the above definition, we define weak
bisimulations via the following lifting of $F{\times}F$:
$$\lift{F \times F} = \rho \otimes \wlift{F \times F}\enspace,$$
where $\rho$ is the constant functor defined as $\rho(R \subseteq X^2) =
\{((U,V), (V,U)) \mid U,V \in FX\}$ and $\wlift{F \times F}$ is the
lax relation lifting of $F\times F$ for the ordering $(U_1, V_1)
[\supseteq \subseteq] (U_2,V_2)$ iff $U_2 \subseteq U_1$ and $V_1
\subseteq V_2$.%, for all $U_i,V_i \in FX$.

For an intuition, observe that an $F \times F$-coalgebra is a pair $\langle \xi_1, \xi_2\rangle \colon X \to FX \times FX$ of LTSs that we denote with $\to_1$ and $\to_2$. An invariant for $\pred{\wlift{F \times F}}{\langle \xi_1,\xi_2 \rangle}$ is a relation $R \subseteq X^2$ such that for each $(x,y) \in R$: (1) if $y \xrightarrow{l}_1 y'$ then $x \xrightarrow{l}_1 x'$ with $x'Ry'$, and (2)
if $x \xrightarrow{l}_2 x'$ then $y \xrightarrow{l}_2 y'$ with $x' R y'$. 
Composing with $\rho$ ``flips'' the LTSs $\to_1$ and $\to_2$:
an invariant for $\pred{\lift{F \times F}}{\langle \xi_1,\xi_2 \rangle}$ is now an $R \subseteq X^2 $ such that:
(1) if $y \xrightarrow{l}_1 y'$ then $x \xrightarrow{l}_2 x'$ with $x'Ry'$, and (2)
if $x \xrightarrow{l}_1 x'$ then $y \xrightarrow{l}_2 y'$ with $x' R y'$.
It is easy to see that for $\langle \xi_1,\xi_2 \rangle= \langle \to, \Ra \rangle$, coalgebras for $\pred{\lift{F \times F}}{\langle \xi_1,\xi_2 \rangle}$ are weak bisimulations and the final coalgebra is weak bisimilarity.

In \appref{app:proof6}, we show that $(\lift{F \times F}, F)$ is a fibration map and by Theorem~\ref{theo:beh} we now obtain the following.
\begin{corollary}
\label{cor:bhv-weak-bisim}
  $\Beh$ is $\pred{\lift{F \times F}}{\langle \xi_1,\xi_2\rangle}$-compatible.
\end{corollary}
For $\langle \xi_1,\xi_2 \rangle= \langle \to, \Ra \rangle$,
behavioural equivalence is simply strong bisimilarity. Consequently,
Corollary~\ref{cor:bhv-weak-bisim} actually gives the compatibility of
weak bisimulation up to strong bisimilarity~\cite{PS12}.
One could wish to use up to $\Slf$ or up to $\Tra$ for weak
bisimulations. However, the condition $(*{*}*)$ from
Section~\ref{ssec:trans:compat} fails, and indeed, weak bisimulations
up to weak bisimilarity or up to transitivity are not
sound~\cite{PS12}.

For up to context, we use Theorem~\ref{thm:comp-lambda} to reduce
compatibility w.r.t. $\lift{F \times F}$ to compatibility
w.r.t. $\rho$ and $\wlift{F \times F}$ (for which we can reuse the
result of the previous section).

\newcommand{\propweak}{
Let $\lambda \colon T(F \times \Id) \Rightarrow F\T$ be a positive GSOS specification and $(X, \alpha,\langle \xi_1, id \rangle)$ and $(X,\alpha,\langle \xi_2, id \rangle)$ be two $\lambda^{\dagger}$-bialgebras then $\Con$ is $\pred{(\lift{F\times F}\times \Id)}{\langle \xi_1,\xi_2, id \rangle}$-compatible.
}
\begin{proposition}\label{prop:weak}
\propweak
\end{proposition}
The above proposition requires both $\to$ and $\Ra$ to be
models~\cite{AFV} of the same positive GSOS specification $\lambda$.
This means that the rules of $\lambda$ should be sound for both $\to$
and $\Ra$. For instance, in the case of CCS, $\Ra$ is not a model of
$\lambda$ because the rule for non-deterministic choice is not sound
for $\Ra$. Nevertheless, we can use our framework to prove the
compatibility of weak bisimulation up to contextual closure w.r.t.\ the
remaining operators.

\section{Directions for future work}

Our nominal automata example leads us to expect that the
framework introduced in this paper will lend itself to obtaining a clean theory of up-to
techniques for name-passing process calculi. For instance, we would
like to understand whether the congruence rule format proposed by
Fiore and Staton~\cite{FioreS09} can fit in our setting: this would
provide general conditions under which up-to techniques related to
name substitution are sound in such calculi.

Another interesting research direction is suggested by the divergence
predicate we studied in Section~\ref{ssec:divergence}. Other formulas
of (coalgebraic) modal logic~\cite{DBLP:journals/cj/CirsteaKPSV11} can
be expressed by taking different predicate liftings, and yield
different families of compatible functors. This suggests a connection
with the proof systems
in~\cite{DBLP:conf/concur/Dam95,Simpson2004287}: we can regard proofs
in those systems as invariants up to some compatible functors.
By using our framework and the logical distributive laws
of~\cite{DBLP:conf/lics/Klin07}, we hope to obtain a systematic way to
derive or enhance such proof systems, starting from a given abstract
GSOS specification.

\acks

We are grateful to the anonymous reviewers for their constructive
comments, and in particular to the one who noticed that our notion of
compatible functor was just an instance of morphisms of endofunctors.
We would also like to thank Alexander Kurz and Alexandra Silva for the stimulating discussions that eventually led to the example on nominal automata; Marcello Bonsangue, Tom Hirschowitz and Henning Kerstan for comments on preliminary versions of the paper; Ichiro Hasuo
for the inspiring talk at Bellairs Workshop on Coalgebras.

%%% Local Variables: 
%%% mode: latex
%%% IspellDict: british
%%% TeX-master: "lics.tex"
%%% End: 

% LocalWords:  dually bifibration Lenisa et al Bartels equivariant modularity
% LocalWords:  determinisation monotonicity isomorphism bicartesian adjoint FX
% LocalWords:  posets adjoints coproducts coproduct tuples fibrewise

\bibliographystyle{plain}

\clearpage
\appendix
\section{Proofs for Section~\ref{sec:compatible-functors}}

The following Proposition generalises the compositionality results for compatible functions on lattices, see~\cite{pous:aplas07:clut} or~\cite[Proposition~6.3.11]{PS12}.

\begin{repproposition}{prop:modularity}
  Compatible functors are closed under the following constructions:
  \begin{enumerate}[(i)]
  \item composition: if $A$ is $(B,C)$-compatible and
    $A'$ is $(C,D)$-compatible, then $A'\circ A$ is
    $(B,D)$-compatible;
  \item pairing: if $(A_i)_{i\in\iota}$ are
    $(B,C)$-compatible, then $\langle A_i\rangle_{i\in\iota}$ is
    $(B,C^\iota)$-compatible;
  \item product: if $A$ is $(B,C)$-compatible and
    $A'$ is $(B',C')$-compatible, then $A\times A'$ is
    $(B{\times}B',C{\times}C')$-compatible; 
  \end{enumerate}
  Moreover, for an endofunctor $B \colon \C\to\C $, 
  \begin{enumerate}[(i)]\setcounter{enumi}{3}
  \item the identity functor $\mathit{Id}\colon \C\to\C$
    is $B$-compatible;
  \item the constant functor to the carrier of any
    $B$-coalgebra is $B$-compatible, in particular the final one if it
    exists;
  \item the coproduct functor $\coprod\colon
    \C^\iota\to\C$ is $(B^\iota,B)$-compatible. %, assuming that $\C$ has coproducts of an ordinal $\iota$.
  \end{enumerate}
\end{repproposition}
\begin{proof}
  \begin{enumerate}[(i)]
  \item Given $\gamma:AB\Ra CA$ and $\gamma':A'C\Ra DA'$ we obtain 
    \begin{equation*}
      \xymatrix{
        A'AB\ar@{=>}[r]^-{A'\gamma}&A'CA\ar@{=>}[r]^-{\gamma' A} &DA'A 
}
    \end{equation*}
  \item Given natural transformations $\gamma_i:A_iB\Ra CA_i$ for
    all $i\in\iota$ we obtain a natural transformation
    \begin{equation*}
      \xymatrix{
        \langle A_i\rangle_{i\in\iota}B\ar@{=}[d] & C^\iota \langle A_i\rangle_{i\in\iota} \ar@{=}[d]\\
        \langle A_iB\rangle_{i\in\iota}\ar@{=>}[r]^-{\langle\gamma_i\rangle_{i\in\iota}} &\langle CA_i\rangle_{i\in\iota}
      }
    \end{equation*}
  \item Given $\gamma:AB\Ra CA$ and $\gamma':A'B'\Ra C'A'$ we
    construct $\gamma\times\gamma':(A\times A')(B\times B')\Ra
    (C\times C')(A\times A')$.
  \end{enumerate}
  \noindent
  Items~\eqref{it:id},~\eqref{it:constant} and~\eqref{it:coproduct}
  are trivial. For example, the latter is immediate using the
  universal property of the coproduct.
\end{proof}

\section{Proofs for Section~\ref{sec:upto:fibrations}}
\label{app:sec:upto:fibrations}
The next simple Lemma about liftings in fibrations will be used throughout this appendix, e.g., to prove   Proposition~\ref{prop:modular-compatibility-with-bisim}, but also Theorem~\ref{thm:big-2-cell}.

\begin{lemma}
\label{eq:2-cell-lifting-T}
  Let $p:\P\to\B$ and $p':\P'\to\B$ be two fibrations and  assume $\lift{T}:\P\to\P'$ is the lifting of a functor $T:\B\to\B$. Consider a $\B$-morphism $f:X\to Y$. Then there exists a natural transformation:
$$\theta:\lift{T}\circ f^* \Ra (Tf)^*\circ\lift{T}:\P_Y\to\P'_{TX}.$$
\end{lemma}
\begin{proof}
 In order to define $\theta_R$ for some $R$ in $\P_Y$, we
  use the universal property of the Cartesian lifting
  $\widetilde{Tf}_{\lift{T}(R)}$. In a diagram:
  \begin{equation}
    \label{eq:embeding-coalg}
    \begin{gathered}
      \begin{tikzpicture}
        \matrix (m) [matrix of math nodes,row sep=3em,column
        sep=4em,minimum width=2em] {
          \lift{T}(f^*(R)) & \\
          (Tf)^*(\lift{T}R) & \lift{T}R  \\
          TX & TY\\
        }; 
        \path[-stealth] 
        (m-1-1) edge node [above,right] {$\lift{T}(\widetilde{f}_R)$} (m-2-2) 
        (m-2-1) edge node [below] {$\widetilde{Tf}_{\lift{T}R}$} (m-2-2) 
        (m-1-1) edge[dashed] node [left] {$\theta_R$} (m-2-1) 
        (m-3-1) edge node [above] {$Tf$} (m-3-2);
      \end{tikzpicture}
    \end{gathered}
  \end{equation}
\end{proof}

\begin{lemma}
  \label{eq:2-cell-lifting-F}
  Let $p:\P\to\B$ be a bifibration and  assume $\lift{F}:\P\to\P$ is the lifting of a functor $F:\B\to\B$. Consider a $\B$-morphism $f:X\to Y$. Then there exists a natural transformation:
$$\rho:\textstyle{\coprod}_{Ff}\circ \lift{F}\Ra\lift{F}\circ\textstyle{\coprod}_f:\P_X\to\P_{FY}.$$
\end{lemma}
\begin{proof}
  The proof uses the universal property of the opcartesian liftings. Equivalently, from Lemma~\ref{eq:2-cell-lifting-T} we have a natural transformation $\lift{F}f^*\Ra(Ff)^*\lift{F}$. Taking the adjoint transpose via $\coprod_{Ff}\dashv(Ff)^*$ we get a natural transformation
$\coprod_{Ff}\lift{F}f^*\Ra\lift{F}$. A further adjoint transpose via the adjunction $\coprod_f\dashv f^*$ yields the desired $\rho:\textstyle{\coprod}_{Ff}\lift{F}\Ra\lift{F}\textstyle{\coprod}_f$.
\end{proof}

\subsection{Proofs for Section~\ref{ssec:trans:compat}}

In this section we prove Proposition~\ref{prop:modular-compatibility-with-bisim}. 
For the sake of clarity we explain how $\lift{F}^n$ is defined for $n=2$. Recall that $\P\times_\B\P$ is obtained as a pullback of $p$ along $p$ in $\Cat$.

For a lifting $\lift{F}$ of $F$, the functor $\lift{F}^2$ makes the next diagram commute.
\begin{equation*}
  \xymatrix@C=0.3cm@R=0.3cm{
\P\times_\B\P\ar[rr]\ar[dd]\ar@{..>}[rd]^-{\lift{F}^2}& &\P\ar[dd]^(0.35){p}\ar[rd]^-{\lift{F}} & \\
& \P\times_\B\P\pullbackcorner\ar[rr]\ar[dd] &\qquad\qquad & \P\ar[dd]^-{p}\\
\P\ar[rr]^(0.3){p}\ar[rd]_-{\lift{F}} & & \B\ar[rd]^-{F} & \\
& \P\ar[rr]_-{p} & & \B \\
}
\end{equation*}
 This means  that on each fibre we have  
$$\lift{F}^n_X=(\lift{F}_X)^n\colon \P_X^n\to\P_{FX}^n.$$

As a consequence of Lemma~\ref{eq:2-cell-lifting-T} we obtain:

\begin{lemma}\label{lem:compatibility-with-cart-funct}
    Let $p:\P\to\B$ and assume $G:\P^{\times_{\B}n}\to \P$ is a lifting  of the identity on $\B$.  If
    $f:X\to Y$ is a $\B$-morphism, there is a canonical natural
    transformation
$$\theta:G(f^*)^n\Ra f^*G:\P_Y\to\P_{GX}.$$
\end{lemma}
\begin{proof}
This is an instance of Lemma~\ref{eq:2-cell-lifting-T} for $T=\Id$ and $\lift{T}=G$. We also use that the Cartesian lifting of a $\B$-morphism $f$ in $\P^{\times_{\B}n}$ is $(f^*)^n$, where $f^*$ is the Cartesian lifting in $\P$.
\end{proof}

\begin{repproposition}{prop:modular-compatibility-with-bisim}
  Let $G\colon \P^{\times_{\B}n}\to \P$ be a lifting of the identity on $\B$ such that there exists  a
  natural transformation
  $G\lift{F}^n\Ra\lift{F}G$. Then $G_X$ is
  $\pred{\lift{F}}{\xi}$-compatible.
\end{repproposition}

\begin{proof}
  Consider the natural transformation obtained as the composition
$$G_X(\xi^*)^n(\lift{F})^n \Ra \xi^* G_X(\lift{F})^n \Ra \xi^*\lift{F} G_X $$
and use that $(\xi^*\circ\lift{F})^n=(\xi^*)^n\circ(\lift{F})^n$. The
first natural transformation comes from
Lemma~\ref{lem:compatibility-with-cart-funct} applied for $\xi$.
\end{proof}

\subsection{Proofs for Section~\ref{ssec:context:compat}}
\label{app:ssec:context:compat}
%\subsection{Lemmas for Theorem \ref{thm:big-2-cell}}
In the next Theorem we only use that the fibration $p:\P\to\B$ is a bifibration and is split. 
\footnote{The Beck-Chevalley condition is not required for the functors $\coprod_f$.}
\begin{reptheorem}{thm:big-2-cell}
  Let $\lift{T},\lift{F}\colon \P\to\P$ be liftings of $T$ and $F$. If
  $\lift{\lambda} \colon \lift{T}\,\lift{F}\Ra\lift{F}\,\lift{T}$ is a
  natural transformation sitting above $\lambda$, then
  $\coprod_\alpha\circ\,\lift{T}$ is $\pred{\lift{F}}{\xi}$-compatible.
\end{reptheorem}

\begin{proof}
  We exhibit a natural transformation
  $$\textstyle{\coprod}_\alpha\circ \lift{T}\circ \xi^*\circ \lift{F} \Rightarrow \xi^*\circ \lift{F}\circ \textstyle{\coprod}_\alpha\circ \lift{T}.$$
  This is achieved in Figure~\ref{eq:compatibility} by pasting five
  natural transformations, obtained as follows:
  \begin{enumerate}[(a)]
  \item is the counit of the adjunction
    $\coprod_{\lambda_X}\dashv\lambda_X^*$.
  \item comes from $\lift{\lambda}$ being a lifting of
    $\lambda$, see Lemma~\ref{eq:2-cell-distributivity}.
  \item comes from the bialgebra condition, the fibration being split,
    and the units and counits of the adjunctions
    $\coprod_{\alpha}\dashv\alpha^*$,
    $\coprod_{F\alpha}\dashv(F\alpha)^*$, and $\coprod_{\lambda_X}\dashv\lambda_X^*$. See Lemma~\ref{eq:2-cell-bialgebra}.
  \item arises since $\lift{T}$ is a lifting of $T$, using the
    universal property of the Cartesian lifting
    $(T\xi)^*$, see Lemma~\ref{eq:2-cell-lifting-T}.
  \item comes from $\lift{F}$ being a lifting of $F$, combined with
    the unit and counit of the adjunction
    $\coprod_{\alpha}\dashv\alpha^*$, see Lemma~\ref{eq:2-cell-lifting-F}. 
    \qedhere
  \end{enumerate}
\end{proof}

\begin{lemma}
  \label{eq:2-cell-distributivity}
Consider a fibration $p:\P\to\B$, two $\B$-endofunctors $F,T$ with corresponding liftings $\lift{T},\lift{F}$. Assume $\lambda:TF\Ra FT$ is a natural transformation and $\lift{\lambda}:\lift{T}\lift{F}\Ra\lift{F}\lift{T}$ sits above $\lambda$. Then there exists a 2-cell as in the diagram below:

\begin{equation}
  \begin{gathered}
    \begin{tikzpicture}
      \matrix (m) [matrix of math nodes,row sep=3em,column
      sep=4em,minimum width=2em] {
        \P_X & \P_{FX} & \P_{TFX} \\
        \P_X &\P_{TX} & \P_{FTX} \\
      }; 
      \path[-stealth] 
      (m-1-1) edge node [above] {$\lift{F}$} (m-1-2)
              edge node [left] {$\id$} (m-2-1)
      (m-1-2) edge node [above] {$\lift{T}$} (m-1-3) 
      (m-2-1) edge node [below] {$\lift{T}$} (m-2-2) 
      (m-2-2) edge node [below] {$\lift{F}$} (m-2-3)
      (m-2-3) edge node [right] {$\lambda_X^*$} (m-1-3)
(m-1-2) edge[color=white]
    node[fill=white] {$\textcolor{black}{\Downarrow}$} (m-2-2);
    \end{tikzpicture}
  \end{gathered}
\end{equation}
\end{lemma}
\begin{proof} For $R\in\P_{FTX}$ the $R$-component of the required natural transformation is the dashed line in the diagram below and is obtained using the universal property of the Cartesian lifting of $\lambda_X$.
    \begin{equation}
  \label{eq:2-cell-distributivity-proof}
    \begin{gathered}
      \begin{tikzpicture}
        \matrix (m) [matrix of math nodes,row sep=3em,column
        sep=4em,minimum width=2em] {
          \lift{T}\lift{F}R & \\
          \lambda^*(\lift{F}\lift{T}R) & \lift{F}\lift{T}R  \\
          TFX & FTX\\
        }; 
        \path[-stealth] 
(m-1-1) edge node [above,right] {$\lift{\lambda}_R$} (m-2-2) 
(m-2-1) edge node [below] {$\widetilde{\lambda}_{\lift{F}\lift{T}R} $} (m-2-2) 
(m-1-1) edge[dashed] node [left] {}(m-2-1) 
(m-3-1) edge node [above] {$\lambda_X$} (m-3-2);
      \end{tikzpicture}
    \end{gathered}
  \end{equation}
The naturality in $R$ can be easily checked and is a consequence of the uniqueness of the factorisation.
\end{proof}

\begin{lemma}
\label{eq:2-cell-bialgebra}  
Given $(X,\alpha,\xi)$ an $\lambda$-bialgebra as in~\eqref{eq:bialg}
\begin{equation}
 \label{eq:bialg}
 \begin{gathered}
   \begin{tikzpicture}
     \matrix (m) [matrix of math nodes,row sep=2em,column
     sep=4em,minimum width=2em] {
       TX & X & FX \\
       TFX & & FTX \\
     }; \path[-stealth] (m-1-1) edge node [above] {$\alpha$} (m-1-2)
     (m-1-1) edge node [left] {$T\xi$} (m-2-1) (m-1-2) edge node
     [above] {$\xi$} (m-1-3) (m-2-3) edge node [right] {$F\alpha$}
     (m-1-3) (m-2-1) edge node [below] {$\lambda_X$} (m-2-3);
   \end{tikzpicture}
 \end{gathered}
\end{equation}
and $p:\P\to\B$ a split fibration, there exists a 2-cell
\begin{equation}
 \label{eq:2-cell-bialgebra-proof}
    \begin{gathered}
    \begin{tikzpicture}
      \matrix (m) [matrix of math nodes,row sep=2.5em,column
      sep=4em,minimum width=2em] {
        \P_{TFX} & \P_{TX} & \P_{X} \\
        \P_{FTX} &\P_{FX} & \P_{X} \\
      }; 
      \path[-stealth] 
      (m-1-1) edge node [above] {$(T\xi)^*$} (m-1-2)
              edge node [left] {$\coprod_{\lambda_X}$} (m-2-1)
      (m-1-2) edge node [above] {$\coprod_{\alpha}$} (m-1-3) 
      (m-2-1) edge node [below] {$\coprod_{F\alpha}$} (m-2-2) 
      (m-2-2) edge node [below] {$\xi^*$} (m-2-3)
      (m-2-3) edge node [right] {$\id$} (m-1-3)
(m-1-2) edge[color=white]
    node[fill=white] {$\textcolor{black}{\Downarrow}$} (m-2-2);
    \end{tikzpicture}
  \end{gathered} 
\end{equation}
\end{lemma}

\begin{proof}
We obtain the required natural transformation as the composite of the natural transformations  of~\eqref{eq:4} below.
  \begin{equation}
    \label{eq:4}
    \begin{array}{cc}
        \coprod_\alpha\circ (T\xi)^* &  \\
& \\
        \Downarrow & (\coprod_\lambda\dashv \lambda^*)\\
& \\
        \coprod_\alpha\circ (T\xi)^*\circ \lambda^*\circ\coprod_{\lambda} & \\
& \\
        \Downarrow & (\coprod_{F\alpha}\dashv (F\alpha)^*)\\
& \\
        \coprod_\alpha\circ (T\xi)^*\circ \lambda^*\circ(F\alpha)^*\circ\coprod_{F\alpha}\circ\coprod_{\lambda} & \\
& \\
        \Downarrow & (\mathrm{bialg})\\
& \\
        \coprod_\alpha\circ \alpha^*\circ\xi^*\circ\coprod_{F\alpha}\circ\coprod_{\lambda} & \\
& \\        
        \Downarrow & (\coprod_{\alpha}\dashv \alpha^*)\\
& \\
       \xi^*\circ\coprod_{F\alpha}\circ\coprod_{\lambda} & \\
    \end{array}      
  \end{equation}

 Except for the third one, these 2-cells are obtained from the units or counits of the adjunctions recalled on the right column. The third natural transformation is actually an isomorphism and arises from  $(X,\alpha,\xi)$ being a bialgebra and the fibration being split.
\end{proof}

\subsection{Proofs for Section~\ref{ssec:abstract-GSOS}}
\label{app:GSOS}

In this section we will prove Lemma~\ref{lm:gsos}. First we recall some basic facts on the free monad $\T$ over a functor $T$ on some category $\C$. 

Assuming $T$ has free algebras over any $X$ in $\C$ one can show that the free monad $\T$ over $T$ exists. We can define $\T X$  as the free $T$-algebra on $X$, or equivalently, as the initial algebra for the functor $X+T(-)$. Thus for each $X$ in $\C$ one has an isomorphism
$$[\eta_X,\kappa_X]\colon X+T\T X\to\T X.$$ 
The $\eta$ above gives the unit of the monad $\T$.
% The composition $\iota = \kappa\circ T\eta\colon T\Ra\T$ satisfies a universal
% property that amounts to saying that $\T$ is the free monad on $T$. 
The monad multiplication $\mu: \T\T X \to \T X$ is given as the unique morphism obtained by equipping $\T X$ with the $\T X+T(-)$-algebra structure $[\id,\kappa_X]$.

Recall from~\cite{DBLP:conf/lics/TuriP97} that there exists a bijective correspondence between natural transformations
  $$\lambda:T(F\times\Id)\Rightarrow F\T$$
and distributive laws
$$\lambda^\dagger\colon\T(F\times\Id)\Rightarrow (F\times\Id)\T. $$
We briefly recall here how $\lambda^{\dagger}$ is obtained from $\lambda$.
For $X$ in $\B$, we equip $(F\times\Id)\T X$ with a $FX\times X + {T}(-)$-algebra structure, given by the sum of :
% $$
% \xymatrix{
% (F\times\Id)X \ar[rrr]^{(F\times\Id)\eta_X} & &
% (F\times\Id)\bSigma X \\
% {T}(F\times\Id)\bSigma X\ar[rr]^{(F\times\Id)\kappa_X\circ\lambda_{\bSigma X}} & & 
% {T}(F\times\Id)\bSigma X
% (F\times\Id)\bSigma X \\
% }
% $$

\begin{equation*}
%  \label{eq:1}
  \xymatrix{
FX\times X\ar[rr]^-{(F\times\Id)\eta_X} & & (F\times\Id)\T X \\
& F\T\T X\ar[r]^-{F\mu_X} & F\T X\\
{T}(F\times\Id)\T X\ar[ur]^-{\lambda_{\T X}}\ar[dr]_-{{T}\pi_2\T_X}\ar@{-->}[rr] & & (F\times\Id)\T X\ar[u]\ar[d] \\
& {T}\T\ar[r]_-{\kappa_X} & \T X \\
}
\end{equation*}
Hence $\ldag_X$ is defined as the unique $(F\times\Id)X + {T}(-)$-algebra morphism:
\begin{equation}
\begin{gathered}
  \label{eq:2}
    \xymatrix{
{T}{\T}(F\times\Id)X
\ar[r]^-{{T}(\ldag_X)}
\ar[d]_-{\kappa_{(F\times\Id)X}} 
& 
{T}(F\times\Id){\T} X
\ar[d]_-{\langle F\mu_X\lambda_{{\T} X}}^-{\kappa_X({T}\pi_2{\T})_X\rangle} 
\\
{\T} (F\times\Id)X
\ar@{-->}[r]^-{\ldag_X} 
& 
(F\times\Id){\T} X
\\
(F\times\Id)X
\ar[u]^-{\eta_{(F\times\Id)X}}
\ar[ur]_-{(F\times\Id)\eta_X}
}
\end{gathered}
\end{equation}

The following technical lemma is needed to establish that whenever the lifting of $\lift{T}$ of a functor $T$ has free algebras, the free monad over $\lift{T}$ is the lifting of the free monad over $T$.

\begin{lemma}\label{lem:free-alg-of-lifting}
  Consider a lifting  $\lift{T}$ of a $\B$-endofunctor ${T}$ and assume $\lift{T}$ has free algebras.
  \begin{enumerate}
  \item The functor $p:\P\to\B$ has a right adjoint $\mathbf{1}:\B\to\P$ inducing an adjunction\footnote{The functor $\Alg$ stems from the 2-categorical notion of \emph{inserter}, see~\cite{Street74} or~\cite[Theorem~2.14,Appendix~A.5]{Hermida97structuralinduction} for a concise exposition.
}
\[
\begin{tikzpicture}
  \matrix (m) [matrix of nodes,row sep=2.5em,column
    sep=3em,minimum width=2em,ampersand replacement=\&] {
\node (1) {$\Alg(\lift{T})$} ;\&
\node {$\bot$} ;\&
\node (2) {$\Alg(T)$}; \\
};
\path[-stealth] (1) edge [bend left=20] node [above] {$\Alg(p)$} (2)
(2) edge [bend left=20] node [below] {$\Alg(\mathbf{1})$} (1);
\end{tikzpicture}
\]
\item The functor $\Alg(p)$ preserves the initial algebras.
\item When $P\in\P_X$ for some $X$ in $\B$, the free $\lift{T}$-algebra  over $P$ sits above the free ${T}$-algebras over $X$.
\item The free monad $\lift{\T}$ over $\lift{T}$ exists and is a lifting of the free monad $\T$ over ${T}$.
\end{enumerate}
\end{lemma}
\begin{proof}
  \begin{enumerate}
\item  Since the fibration considered here is bicartesian, one can define $\mathbf{1}(X)$ as the terminal object in $\P_X$. Then the statement of this item is an immediate consequence of~\cite[Theorem~2.14]{Hermida97structuralinduction}. 

\item  follows because $\Alg(p)$ is a left adjoint. 

\item follows from item 1) applied for the lifting $P+\lift{T}$ of $X+{T}$.
\item is an immediate consequence of item 3).
\end{enumerate}
\end{proof}

\begin{replemma}{lm:gsos}
  Consider a lifting  $\lift{T}$ of a $\B$-endofunctor ${T}$ and assume $\lift{T}$ has free algebras. Let $\lift{\lambda}:\lift{T}(\lift{F}\times\lift{\Id})\Rightarrow \lift{F}\lift{\T}$ be a natural transformation sitting above 
  $\lambda:T(F\times\Id)\Rightarrow F\T$.
Then $ \lift{\lambda}^\dagger \colon\lift{\T}(\lift{F}\times\lift{\Id})\Rightarrow (\lift{F}\times\lift{\Id})\lift{\T} $ sits above $\lambda^\dagger\colon\T(F\times\Id)\Rightarrow (F\times\Id)\T $.
\end{replemma}

\begin{proof}%[Proof of Lemma \ref{lm:gsos}]
We know that ${\T} X$ is the free ${T}$-algebra on $X$. Let 
$$[\eta_X,\kappa_X]:X+{T}{\T} X\to{\T} X$$ denote the initial $X+{T}(-)$-algebra.
Similarly, let 
$$[\lift{\eta}_P,\lift{\kappa}_P]:P+\lift{T}\lift{{\T}} P\to\lift{{\T}} P$$
denote the initial $P+\lift{T}(-)$-algebra. By Lemma~\ref{lem:free-alg-of-lifting} we know that when $P\in\P_X$ we have that $[\lift{\eta}_P,\lift{\kappa}_P]$ is a lifting of $[\eta_X,\kappa_X]$.

For $P\in\P_X$ the map $\ldagbar_P$ is defined similarly to~\eqref{eq:2}, as the unique map such that:
\begin{equation}
  \label{eq:3}
\begin{gathered}
    \xymatrix{
\lift{T}\lift{{\T}}(\lift{F}\times\lift{\Id})P
\ar[r]^-{\lift{T}(\ldagbar_P)}
\ar[d]_-{\lift{\kappa}_{(\lift{F}\times\lift{\Id})P}} 
& 
\lift{T}(\lift{F}\times\lift{\Id})\lift{{\T}} P
\ar[d]_-{\langle \lift{F}\lift{\mu}_P\lift{\lambda}_{\lift{{\T}} P}}^-{\lift{\kappa}_P(\T\pi_2\lift{{\T}})_P\rangle} 
\\
\lift{{\T}} (\lift{F}\times\lift{\Id})P
\ar@{-->}[r]^-{\ldagbar_P} 
& 
(\lift{F}\times\lift{\Id})\lift{{\T}} P
\\
(\lift{F}\times\lift{\Id})P
\ar[u]^-{\lift{\eta}_{(\lift{F}\times\lift{\Id})P}}
\ar[ur]_-{(\lift{F}\times\lift{\Id})\lift{\eta}_P}
}
\end{gathered}
\end{equation}
By Lemma~\ref{lem:free-alg-of-lifting} we have that the $(\lift{F}\times\lift{\Id})P + \lift{T}(-)$-algebras  $\lift{{\T}} (\lift{F}\times\lift{\Id})P$
and $(\lift{F}\times\lift{\Id})\lift{{\T}} P$ of diagram~\eqref{eq:3} sit above the 
 $(F\times\Id)X + {T}(-)$-algebras $\T(F\times\Id) X$,
respectively  $(F\times\Id){\T} X$ of diagram~\eqref{eq:2}. By uniqueness of $\ldag_X$ it follows that $\ldagbar_P$ sits above $\ldag_X$.
\end{proof}

\section{Details on Divergence}\label{app:divergence}
In this appendix, we discuss some details for showing compatibility of $\Con^{\ell}$ that were omitted in the main text for lack of space.

First of all, observe that the GSOS rules defining the parallel operator corresponds to a distributive law $\lambda\colon T(F\times \Id) \Rightarrow F\T$, which is defined for all sets $X$, $x,y\in X$ and $S,T\in \pow(L\times X)$ as
\begin{align*}
(S,x), (T,y) \mapsto &~ \{(l,x'|y) \mid  (l,x') \in S\} \\ 
\cup &~ \{(l,x|y') \mid  (l,y') \in T\} \\
\cup &~ \{ (\tau, x'|y') \mid  \exists a,  (a,x') \in S \land (\bar{a},y') \in T \}\\   
\cup &~ \{ (\tau, x'|y') \mid  \exists a,  (\bar{a},x') \in S \land (a,y') \in T \}\text{.}
\end{align*}
Intuitively, $S$ and $T$ are the sets of transitions of the states $x$
and $y$. The first set $\{(l,x'|y) \mid (l,x') \in S\}$ corresponds to
the first GSOS rule 
\begin{align*}
  \frac{x\tr{l}x'}{x|y\tr{l}x'|y}
\end{align*}
and similarly for the others.

By virtue of Lemma \ref{lm:gsos}, to prove compatibility of $\Con^{\ell}$, we only have to show that for all predicates $P\subseteq X$, the restriction of $\lambda_X$ to $\lift{T}(\lift{F}^{\diamondtau} \times \Id )P$ corestricts to $\lift{F}^{\diamondtau}\lift{\T} P$, that is whenever $(S,x), (T,y)\in \lift{T}( \lift{F}^{\diamondtau} \times \Id )P$, then $\lambda_X ((S,x), (T,y)) \in \lift{F}^{\diamondtau}\lift{\T} P$. 

The latter means, by definition of $\lift{F}^{\diamondtau}$, that there exists a $(\tau, t)\in \lambda_X ((S,x), (T,y))$ such that $t \in \lift{\T} P$.  This can be proved as follows: since $S\in \lift{F}^{\diamondtau}P$, then there exists $(\tau,x')\in S$ such that $x'\in P$. By definition of $\lambda_X$, $(\tau,x'|y) \in \lambda_X ((S,x), (T,y))$. Finally, since $x'\in P$, then $x'|y 
\in \lift{\T} P$.

\section{Details on Nominal Automata}
\label{ssec:nom-aut-ex}

In this section we assume the reader has some familiarity with nominal
sets, see~\cite{Pitts-book}.
\subsection{The base category}

We denote by $\At$ a countable set of names. The category $\Nom$ of
nominal sets has as objects sets $X$ equipped with an action
$\cdot:\Sym(\At)\times X\to X$ of the group of finitely supported
permutations on $\A$ (that is, permutations generated by transpositions
of the form $(a\ b)$) and such that each $x\in X$ has a finite
support. Morphisms in $\Nom$ are \emph{equivariant} functions, i.e., 
functions that preserve the group action.

\subsection{The fibration at issue}

It is well known that $\Nom$ can equivalently be described as a
Grothendieck topos. Since $\Nom$ is a regular category,
by~\cite[Observation~4.4.1]{Jacobs:fib} we know that the subobject
fibration on $\Nom$ is in fact a bifibration. Furthermore, by a
change-of-base situation described below we obtain the bifibration
$\Rel(\Nom)\to\Nom$, see also~\cite[Example~9.2.5(ii)]{Jacobs:fib}
$$
\xymatrix{
  \Rel(\Nom)\ar[r]\ar[d]&\Sub(\Nom)\ar[d] \\
  \Nom\ar[r]^{I\mapsto I\times I} & \Nom }
$$
Objects of $\Rel(\Nom)$ are equivariant relations. That is, if $X$ is
a nominal set, a nominal relation on $X$ is just a subset $R\subseteq
X^2$ such that $x R y$ implies $(\pi\cdot x) R(\pi\cdot y)$ for all
permutations $\pi$. This bifibration is also split and bicartesian.

\subsection{The functors and the distributive law}

We will use the following $\Nom$-endofunctors:
\begin{enumerate}
\item $F:\Nom\to\Nom$ given by $FX=2\times X^\At$, where $2=\{0,1\}$
  is equipped with the trivial action and $X^\At$ is given by the
  internal hom. Concretely, an element $f\in X^\At$ is a function
  $f:\At\to X$ such that there exists a finite subset $S\subseteq \At$
  and $f(\pi(a))=\pi\cdot f(a)$ for all names $a\in\At$ and
  permutations $\pi\in\Sym(\At)$ fixing the elements of $S$.

\item $\Pf:\Nom\to\Nom$ that maps a nominal set $X$ to its
  orbit-finite finitely supported subsets.  In particular one can
  check that $\Pf$ is a monad and let $\mu$ denote its multiplication,
  given by union.
\end{enumerate}

The functors $\Pf$ and $F$ are related by a distributive law
$$\lambda:\Pf F\Ra F\Pf.$$
For a nominal set $X$, the map $\lambda_X$ is given by the product of
the morphisms acting on $S\in\Pf F(X)$ by
$$ S \mapsto  1\in 2 \mathrm{\ iff\ } 1\in(\Pf \tau_1)(S)$$
and
$$  S \mapsto \lambda a.\{x\in X | \exists f\in (\Pf\tau_2)(S).\  f(a)=x\}\in (\Pf X)^A$$
where $\tau_1,\tau_2$ are the projections from $FX$ to $2$,
respectively $X^\At$.

\subsection{The liftings}

The distributive law $\lambda$ can be lifted to $\Rel(\Nom)$,
see~\cite[Exercise 4.4.6]{Jacobs:coalg}.
$$\Rel(\lambda):\Rel(\Pf) \Rel(F)\Ra \Rel(F)\Rel(\Pf).$$
Concretely, for $R\in\Rel(\Nom)_X$, the nominal relation $\Rel(F)(R)$
is given by $(o,f)\ \Rel(F)(R)\ (o',f')$ iff $o=o'$ and for all
$a\in\At$ we have $f(a) Rf'(a)$.

On the other hand $\Rel(\Pf)$ is given by $S\ \Rel(\Pf)(R)\ S'$ iff
for all $x\in S$ exists $y\in S'$ with $xR y$ and for all $y\in S'$
exists $x\in S$ with $xR y$.  As for $\Rel(\lambda)_R$, this is
obtained as the restriction of $\lambda_R\times \lambda_R$ to
$\Rel(\Pf) \Rel(F)(R)$.
% The incredulous reader can check the details :)

\subsection{Soundness of bisimulation up to congruence}

Nondeterministic nominal automata~\cite{BojanczykKL11} can be modelled
as $F\Pf$-coalgebras, while deterministic nominal automata are
represented as $F$-coalgebras. The classical notion of finiteness is
replaced by orbit-finiteness---from a categorical perspective this
makes sense, since orbit-finite nominal sets are exactly the finitely
presentable objects in the lfp category $\Nom$.

The generalised powerset construction~\cite{SilvaBBR10} can be applied
in this situation as well, that is, a nondeterministic nominal
automata modelled as a coalgebra
$$\langle o,t \rangle:X\to 2\times\Pf(X)^A $$ 
yields an $F$-coalgebra structure
$$\langle o^\sharp,t^\sharp \rangle:\Pf X\to 2\times (\Pf X)^A,$$
on $\Pf X$, given by the composite $F(\mu) \circ \lambda \circ \Pf(\langle o,t
\rangle)$.  The reason why determinisation fails in a nominal
setting~\cite{BojanczykKL11} is that the finitary power object functor
$\Pf$ does not preserve orbit finiteness. This is the case in the
example of Section~\ref{ssec:nominal-automata}.

Notice that $(\Pf X, \mu, \langle o^\sharp,t^\sharp \rangle)$ is a
$\lambda$-bialgebra.

The fibrations $\Rel(\Nom)\to\Nom$ and $\Sub(\Nom)\to\Nom$ are
well-founded in the sense of~\cite{HasuoCKJ:coindFib}.
To prove this we can apply~\cite[Lemma~3.4]{HasuoCKJ:coindFib}, which
gives as a sufficient condition for well-foundedness: that the fibre above each finitely presentable object be finite. Indeed, recall from~\cite{Petrisan:phd} that finitely presentable nominal sets are the orbit-finite ones. Then, it is easy to check that a nominal set with $n$ orbits has $2^n$ \emph{equivariant} nominal subsets.

Hence, by~[Theorem~3.7]\cite{HasuoCKJ:coindFib}, the final
$\cpred{F}{\langle o,t \rangle}$-coalgebra exists and can be computed
as the limit of an $\omega^\op$-chain in the fibre $\Rel(\Nom)_X$. 
We will use this
coinductive predicate to prove that two states of a nominal automata
accept the same language.

We can apply Theorem~\ref{thm:big-2-cell} to prove that the contextual
closure $\Con=\coprod_\mu\circ\Rel(\Pf)$ is $\cpred{F}{\langle o^{\sharp},t^{\sharp}
  \rangle}$-compatible.

Thus bisimulation up to context is a valid proof technique for nominal
automata.

Moreover, we can apply Proposition~\ref{prop:modular-compatibility-with-bisim}
to prove compatibility of the up to reflexive, symmetric and
transitive closure techniques, respectively.
\begin{enumerate}
\item[$(n{=}0)$] Let $\Ref\colon \Nom\to\Rel(\Nom)$ be the functor
  mapping each nominal set $X$ to $\Delta_X$, the identity relation on
  $X$. Then $\Ref_X$ is $\pred{\Rel(F)}{\langle o,t
    \rangle}$-compatible since $\Delta_{FX}= \Rel(F)\Delta_X$.
\item[$(n{=}1)$] Let $\Sym \colon \Rel(\Nom)\to\Rel(\Nom)$ be the
  functor mapping each nominal relation $R\subseteq X^2$ to its
  converse $R^{-1}\subseteq X^2$. $\Sym_X$ is $\pred{\lift{F}}{\langle
    o,t \rangle}$-compatible since $\lift F(R)^{-1} \subseteq \lift
  F(R^{-1})$ for all relations $R\subseteq X^2$.
\item[$(n{=}2)$] Let $\otimes\colon
  \Rel(\Nom)\times_\Nom\Rel(\Nom)\to\Rel(\Nom)$ be the nominal
  relational composition functor. Composition of nominal relations is
  computed just as in $\Set$ and one can show that $\Rel(F)$ preserves
  it. Thus $\otimes$ is $\pred{\Rel(F)}{\langle o,t
    \rangle}$-compatible.
\end{enumerate}

Employing Proposition~\ref{prop:modularity} and the fact that
congruence closure is obtained as the composition of the equivalence,
context and reflexive closure functors we derive that bisimulation up
to congruence is a sound technique.

\subsection{The concrete example}

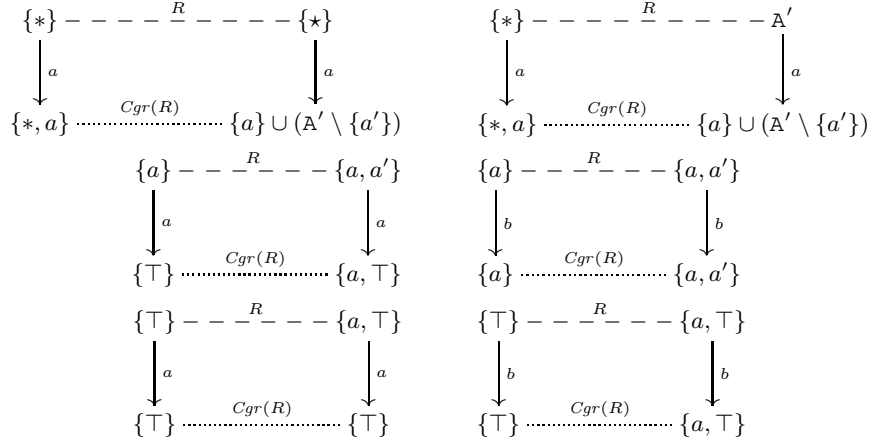
\begin{figure*}[t]
  \centering $\begin{array}{rcl} 
  \xymatrix{
      \{*\}\ar[d]^{a}\ar@{--}[rr]^-{R} & & \{\star\}\ar[d]^{a} \\
      \{*,a\}\ar@{.}[rr]^-{\Cgr(R)} & &\{a\}\cup (\At'\setminus\{a'\})
    } & \ & 
         \xymatrix{
      \{*\}\ar[d]^{a}\ar@{--}[rr]^-{R} & & \At'\ar[d]^{a} \\
      \{*,a\}\ar@{.}[rr]^-{\Cgr(R)} & &\{a\}\cup (\At'\setminus\{a'\})
    } 
    \\
    \xymatrix{
      \{a\}\ar[d]^{a}\ar@{--}[rr]^-{R} & & \{a,a'\}\ar[d]^{a} \\
      \{\top\}\ar@{.}[rr]^-{\Cgr(R)} & &\{a,\top\} }   & & 
       \xymatrix{
      \{a\}\ar[d]^{b}\ar@{--}[rr]^-{R} & & \{a,a'\}\ar[d]^{b} \\
      \{a\}\ar@{.}[rr]^-{\Cgr(R)} & &\{a,a'\} }  
\\    
   \xymatrix{
      \{\top\}\ar[d]^{a}\ar@{--}[rr]^-{R} & & \{a,\top\}\ar[d]^{a} \\
      \{\top\}\ar@{.}[rr]^-{\Cgr(R)} & &\{\top\} }
  &&
     \xymatrix{
      \{\top\}\ar[d]^{b}\ar@{--}[rr]^-{R} & & \{a,\top\}\ar[d]^{b} \\
      \{\top\}\ar@{.}[rr]^-{\Cgr(R)} & &\{a,\top\} }
 \end{array}
  $
  \caption{Proving $R$ to be a bisimulation up to congruence}
  \label{fig:nom-aut-ex}
\end{figure*}

The nondeterministic nominal automaton of
Section~\ref{ssec:nominal-automata} (reported on the left below) is given formally by an
$F\Pf$-coalgebra $\langle o,t \rangle$ on the nominal set
$1+1+\At+\At+1$. For simplicity we denote the second copy of
$\At$ by $\At'$. The map $\langle o,t \rangle$ is given below on the right.
$$
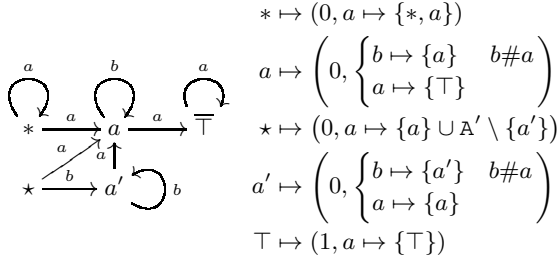
\begin{array}{cc}
 \nfa{\xymatrix @C=.8em @R=1em { %
      \state{{*}}\ar@(ul,ur)^-{a}  \ar[rr]^-{a}&  & %
      \state{a}  \ar@(ul,ur)^-{b}\ar[rr]^-{a} & & \state{\bar{\top}}\ar@(ul,ur)^{a}\\
      \state{\star}\ar[rru]^-{a}\ar[rr]^-{b} & & \state{a'}\ar@(rd,ru)_-{b}\ar[u]^-{a} & &
    }}
&
\begin{aligned}
  * & \mapsto \left(0,a\mapsto \{*,a\}\right)  \\
  a & \mapsto \left(0,
    \begin{cases}
      b\mapsto \{a\} & b\# a\\
      a\mapsto \{\top\} 
    \end{cases}
  \right)\\
  \star & \mapsto \left (0, a\mapsto\{a\}\cup\At'\setminus\{a'\}\right)\\
  a' & \mapsto \left (0,
    \begin{cases}
      b\mapsto \{a'\} & b\# a\\
      a\mapsto \{a\}
    \end{cases}
  \right)\\
  \top & \mapsto \left (1, a\mapsto\{\top\}\right)
\end{aligned}
\end{array}
$$

The determinisation of this automaton has infinitely many orbits. For
example, the determinisation of the part reachable from $*$ is
partially represented by
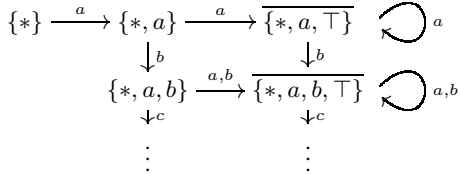
\begin{equation*}
  \nfa{\xymatrix @C=-2pt @R=1em { %
      \state{\{*\}}\ar[rrrrrr]^-{a}&&&&  && 
      \state{\{*,a\}}  \ar[d]^-{b}\ar[rrrrrr]^-{a} &&&& && \state{\bar{\{*,a,\top\}}}\ar[d]^-{b} & {}\ar@(ru,rd)^{a}\\
      &&&& && \state{\{*,a,b\}}\ar[rrrrrr]^-{a,b}\ar[d]^-{c} &&&&& & \state{\bar{\{*,a,b,\top\}}}\ar[d]^-{c} & {}\ar@(ru,rd)^{a,b}\\
      &&&& && \vdots &&&& && \vdots & \\
    }}
\end{equation*}

However, we can prove that $*$ and $\star$ accept the same language, showing that the nominal relation $R$ spanned by
$$(\{*\},\{\star\}),~ (\{a\},\{a,a'\}),~ (\{\top\},\{ a,\top\}),~ (\{*\},\A') $$
is a bisimulation up to congruence, that is, $R\subseteq \cpred{F}{\langle o^{\sharp},t^{\sharp}
  \rangle} \Cgr(R)$.

This is shown in Figure~\ref{fig:nom-aut-ex}: for each pair in $R$, we check that the successors are in $\Cgr(R)$. 
Note that for the pairs $(\{a\},\{a,a'\})$ and $(\{\top\},\{ a,\top\})$, in the second and third rows, one needs to check the successors for $a$ and for a fresh name $b$. Instead for the pairs $(\{*\},\{\star\})$ and $(\{*\},\A')$ in the first row, only successors for $a$ should be checked (since $a$ does not belong to the support of these states).

The only non-trivial computation is to check whether $\{*,a\} \Cgr(R) \{a\}\cup(\At'\setminus\{a'\})$. We proceed as follows: 
\begin{equation*}
  \begin{array}{lcl}
    \{*,a\} & \Cgr(R) & \{a\}\cup\At' \\
    & \Cgr(R) &\{a,a'\}\cup(\At'\setminus\{a'\})\\
    & \Cgr(R) & \{a\}\cup(\At'\setminus\{a'\}).\\
  \end{array}
\end{equation*}

\section{Proofs for Section \ref{sec:compositional}}\label{app:proof6}

\begin{reptheorem}{thm:comp-lambda}
 \thmcomplambda
\end{reptheorem}

\begin{proof}
  Since $\lift{F_1}$ and $\lift{F_2}$ are liftings of $F:\B\to\B$ it follows that $\langle \lift{F_1}, \lift{F_2}\rangle:\P\to\P\times_\B\P$ is a lifting of $F$. Moreover $\langle \lambda_1,\lambda_2\rangle:\lift{T}^2\langle\lift{F_1},\lift{F_2}\rangle\Ra\langle\lift{F_1},\lift{F_2}\rangle\lift{T}$ is a lifting of $\lambda$.

Using that $\otimes:\P\times_\B\P\to\P$ lifts the identity we get that $F_1\otimes F_2=\otimes\circ\langle F_1,F_2\rangle$ is also a lifting of $F$. 

\begin{equation}
\label{eq:comp-lam}
\xymatrix@R=1.5em{
\lift{T}\otimes\langle F_1,F_2\rangle\ar@{=>}[r]^-{\gamma\langle F_1,F_2\rangle} &
\otimes\lift{T}^2\langle F_1,F_2\rangle\ar@{=>}[r]^-{\otimes\langle \lambda_1,\lambda_2\rangle} &
\otimes\langle F_1,F_2\rangle\lift{T}\\
TF\ar@{=>}[r]^-{\id} &
TF\ar@{=>}[r]^-{\lambda} &
FT\\
}
\end{equation}
The required $\lift{\lambda}$ is obtained as the composite
$\otimes\langle \lambda_1,\lambda_2\rangle\circ \gamma\langle F_1,F_2\rangle $ sitting above $\lambda$ as in~\eqref{eq:comp-lam}.
\end{proof}

\subsection{Proofs for Similarity}\label{app:similarity}

\begin{repproposition}{prop:sim}
 \propsim
\end{repproposition}
\begin{proof}
 Recall that $\Con$ is defined as $\textstyle{\coprod}_\alpha\circ \Rel(\T)$ and that, for the canonical lifting, it holds that   $\Rel(\T)\otimes\subseteq \otimes\Rel(\T)^2$. We decompose the lifting $\laxlift{F}\times \Id$ as 
$$(\lift{\sqsubseteq} \times \lift{\Id}) \otimes ( \Rel(F)\times \Id ) \otimes (\lift{\sqsubseteq} \times \lift{\Id})$$
where $\lift{\Id}$ is the constant functor mapping $R\subseteq X^2$ to
$\Delta_X$. By Theorem~\ref{thm:comp-lambda} we reduce the proof of the fact that $\Rel(T)$ distributes over  $\laxlift{F}\times \Id$ to the fact that $\Rel(T)$ distributes over  $\lift{\sqsubseteq} \times \lift{\Id}$ and $\Rel(F)\times \Id$ separately.

%  We use Theorem~\ref{thm:comp-lambda}, to reduce
% compatibility w.r.t.\ $(\laxlift{F}\times \Id)$ to compatibility
% w.r.t.\ $\lift{\sqsubseteq} \times \lift{\Id}$ and $\Rel(F)\times
% \Id$ separately.

For the latter, observe that $\Rel(F)\times \Id = \Rel(F\times \Id)$. Since $\Rel(-)$ is a 2-functor~\cite[Exercise 4.4.6]{Jacobs:coalg}, we take $\lift{\lambda^{\dagger}_1} \colon \Rel(\T)\Rel(F\times \Id) \Ra \Rel(F\times \Id)\Rel(T)$ as $\Rel(\lambda)$.

For the former we need to use Lemma~\ref{lm:gsos} and exhibit a $\lift{\lambda} \colon \Rel(T)(\lift{\sqsubseteq} \times \lift{\Id}) \Ra \lift{\sqsubseteq} \Rel(\T)$ sitting above $\lambda$. This amounts to show that, for all relations $R\subseteq X^2$, the restriction of $\lambda_X \times \lambda_X$ to $\Rel(T)(\lift{\sqsubseteq} \times \lift{\Id})R$ corestricts to $\lift{\sqsubseteq} \Rel(\T)R$. Note that since $\lift{\sqsubseteq}$ and  $\lift{\Id}$ are constant, this is exactly the condition for monotone abstract GSOS. This guarantees the existence of $\lift{\lambda_2^{\dagger}} \colon \Rel(\T)(\lift{\sqsubseteq} \times \lift{\Id}) \Ra (\lift{\sqsubseteq} \times \lift{\Id}) \Rel(\T)$ sitting above $\lambda^{\dagger}$.

The existence of $\lift{\lambda_1^{\dagger}}$ and $\lift{\lambda_2^{\dagger}}$ ensures, via Theorems~\ref{thm:comp-lambda} and~\ref{thm:big-2-cell}, that $\Con$ is $\pred{(\laxlift{F}\times \Id)}{\langle \xi, id\rangle}$-compatible.
\end{proof}

\subsection{Proofs for Weak Bisimilarity}\label{app:weak}

\begin{lemma}
  $(\lift{F \times F}, F)$ is a fibration map.
\end{lemma}
\begin{proof}
  Let $f: X \rightarrow Y$ be a function and $R\subseteq X^2$ be a relation. Then
  \begin{align*}
%    & \rho \otimes [\supseteq \subseteq] \otimes \Rel(F \times F)((f \times f)^{-1}(R)) \otimes [\supseteq \subseteq] \\
    & \lift{F \times F} ((f \times f^{-1}(R)) \\
   &= \{(S,U,V,W) \mid \\ 
   & \qquad \begin{array}{ll}
      \forall (a,x) \in S.~\exists (a,y) \in W.~f(x) R f(y),\\
      \forall (a,y) \in V.~\exists (a,x) \in U.~f(x) R f(y) \}
    \end{array} \\
   & = \{(S,U,V,W) \mid \\
      & \qquad \begin{array}{ll}
      \forall (a,x') \in Ff[S].~\exists (a,y') \in Ff[W].~x' R y',\\
      \forall (a,y') \in Ff[V].~\exists (a,x') \in Ff[U].~x' R y' \}
    \end{array} \\
%   & = (Ff \times Ff \times Ff \times Ff)^{-1}(\rho \otimes [\supseteq \subseteq] \otimes \Rel(F \times F)(R) \otimes [\supseteq \subseteq]) 
   & = (Ff \times Ff \times Ff \times Ff)^{-1}(\lift{F \times F}(R))
  \end{align*}
\end{proof}

\begin{repproposition}{prop:weak}
\propweak
\end{repproposition}
\begin{proof}
From $\lambda\colon T(F \times \Id) \Rightarrow F\T$, we define $\tilde{\lambda}\colon T(F \times F \times \Id) \Rightarrow (F\times F)\T$
as $\langle \lambda\circ T(\tau_1 \times \tau_3)  , \lambda\circ T(\tau_2 \times \tau_3) \rangle $ where $\tau_i$ are the projections from $F\times F\times \Id$ to $F$ and $\Id$. Such $\tilde{\lambda}$ induces a distributive law $$\tilde{\lambda}^{\dagger}\colon \T(F \times F \times \Id) \Rightarrow (F\times F \times \Id)\T\text{.}$$
From the $\lambda^{\dagger}$-bialgebras $(X, \alpha,\langle \xi_1, id \rangle)$ and $(X,\alpha,\langle \xi_2, id \rangle)$, we construct 
$(X, \alpha, \langle \xi_1,\xi_2,id \rangle)$ which is a $\tilde{\lambda}^{\dagger}$-bialgebra.

Recall that $\Con$ is defined as $\textstyle{\coprod}_\alpha\circ \Rel(\T)$ and that, for the canonical lifting, it holds that   $\Rel(\T)\otimes\subseteq \otimes\Rel(\T)^2$. We decompose the lifting $\lift{F\times F}\times \Id$ as 
$$(\rho \times \lift{\Id}) \otimes ( \wlift{F \times F} \times \Id )$$
where $\lift{\Id}$ is the constant functor mapping $R\subseteq X^2$ to $\Delta_X$. 
% We use Theorem~\ref{thm:comp-lambda}, to reduce compatibility of $(F \times F \times \Id)$ to compatibility of $\rho \times \lift{\Id}$ and $\wlift{F \times F}\times \Id$.
By Theorem~\ref{thm:comp-lambda} we reduce the proof of the fact that $\Rel(T)$ distributes over  $\lift{F\times F}\times \Id$ to the fact that $\Rel(T)$ distributes over $\rho \times \lift{\Id}$ and $\wlift{F \times F}\times \Id$ separately.

For the former, by Lemma \ref{lm:gsos}, we have to prove that for all relations $R\subseteq X^2$, the restriction of $\tilde{\lambda}_X \times \tilde{\lambda}_X$ to $\Rel(T)(\rho \times \lift{\Id})R$ corestricts to $\rho \Rel(\T)R$. This can be easily checked by using the fact that both $\rho$ and $\lift{\Id}$ are constant and exploiting the definition of $\tilde{\lambda}$. As a consequence there exists a 
$\lift{\tilde{\lambda_1}}^{\dagger}\colon \Rel(\T)( \rho \times \lift{\Id} ) \Rightarrow (\rho \times \lift{\Id})\Rel(\T)$
sitting above $\tilde{\lambda}^{\dagger}$.

For  $\wlift{F \times F} \times \Id $ we can reuse Proposition \ref{prop:sim}, but first we have to prove that the GSOS specification $\tilde{\lambda}$ is monotone w.r.t. $[\supseteq \subseteq]$. Via simple computations, one can check that this is indeed the case when the original GSOS specification $\lambda$ is positive. As a consequence there exists a 
$\lift{\tilde{\lambda_2}}^{\dagger}\colon \Rel(\T)( \wlift{F \times F} \times \Id ) \Rightarrow (\wlift{F \times F} \times \Id)\Rel(\T)$
sitting above $\tilde{\lambda}^{\dagger}$.

The existence of $\lift{\tilde{\lambda_1}}^{\dagger}$ and $\lift{\tilde{\lambda_2}}^{\dagger}$ entails, via Theorems \ref{thm:comp-lambda} and \ref{thm:big-2-cell} compatibility of $\Con$ for $\pred{(\lift{F\times F}\times \Id)}{\langle \xi_1,\xi_2, id \rangle}$.
\end{proof}

%%% Local Variables: 
%%% mode: latex
%%% IspellDict: british
%%% TeX-master: "arxiv.tex"
%%% End: 

% LocalWords:  dually bifibration Lenisa et al Bartels equivariant modularity
% LocalWords:  determinisation monotonicity isomorphism bicartesian adjoint FX
% LocalWords:  posets adjoints coproducts coproduct tuples fibrewise

\newpage
\renewcommand\refname{Additional references for the appendix}

\end{document}